\documentclass[11pt]{article}
\usepackage[usenames,dvipsnames]{color}
\usepackage{amssymb,amsmath,graphics,color,enumerate}
\usepackage{tikz}
\usepackage{xspace}
\usepackage{todonotes,hyperref}
\usepackage[utf8x,utf8]{inputenc} 
\usepackage[T1]{fontenc}
\usepackage[vlined,linesnumbered,boxruled]{algorithm2e}
\usepackage[left=1in,right=1in,top=1in,bottom=1in]{geometry}
\usepackage[absolute]{textpos}

\usepackage{enumitem}
\usepackage{xspace}
\usetikzlibrary{decorations.pathreplacing,decorations.pathmorphing,patterns,calc,arrows}
\usepackage{todonotes}

\usepackage{comment}

\usepackage{epsfig}
\usepackage{subfig}
\usepackage{float}
\usepackage{amsmath,amsfonts,amssymb,epsfig,color,amsthm}

\usepackage{thmtools, thm-restate}
\declaretheorem{theorem}

\newtheorem{lemma}[theorem]{Lemma}
\newtheorem{corollary}[theorem]{Corollary}

\newtheorem{definition}[theorem]{Definition}
\newcommand{\ignore}[1]{}%

\newcommand{\ProblemFormat}[1]{{\sc #1}}
\newcommand{\ProblemName}[1]{\ProblemFormat{#1}\xspace}

\newcounter{rulecnt}

\newcommand{\defproblempar}[4]{
	\vspace{1mm plus 2mm minus 1mm}
	\noindent\fbox{\begin{minipage}{0.97\textwidth}
		#1\\
		\textbf{Input:} #2\\
		\textbf{Parameter:} #3\\
		\textbf{Question:} #4
  \end{minipage}}
  \vspace{1mm plus 2mm minus 1mm}
}

\newcommand{\Ff}{{\ensuremath{\mathcal{F}}}}

\newcommand{\Oh}{\ensuremath{\mathcal{O}}}

\begin{document}
\title{On Directed Feedback Vertex Set parameterized by treewidth\footnote{Work supported by the National Science Centre of Poland, grants number 2013/11/D/ST6/03073 (MP, MW). The work of {\L}. Kowalik is a part of the project TOTAL that has received funding from the European Research Council (ERC) under the European Union’s Horizon 2020 research and innovation programme (grant agreement No 677651). This research is a part of projects that have received funding from the European Research Council (ERC) under the European Union's Horizon 2020 research and innovation programme under grant agreements No 714704 (AS). MP and MW are supported by the Foundation for Polish Science (FNP) via the START stipend programme.}} 


\author{Marthe Bonamy\thanks{CNRS, LaBRI, France, \texttt{marthe.bonamy@u-bordeaux.fr}} \and \L ukasz Kowalik\thanks{University of Warsaw, Poland, \texttt{\{kowalik,michal.pilipczuk,as277575,m.wrochna\}@mimuw.edu.pl}} \and Jesper Nederlof\thanks{Eindhoven University of Technology, Netherlands,
  \texttt{j.nederlof@tue.nl}} \and Micha\l \ Pilipczuk\footnotemark[3] \and Arkadiusz Soca\l a\footnotemark[3] \and Marcin Wrochna\footnotemark[3]}


\date{}

\begin{titlepage}
\maketitle
\begin{textblock}{20}(0, 12.5)
	\includegraphics[width=40px]{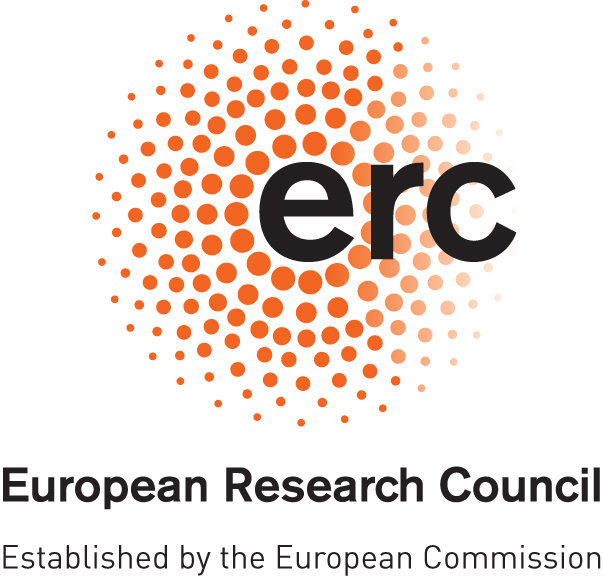}%
\end{textblock}
\begin{textblock}{20}(-0.25, 12.9)
	\includegraphics[width=60px]{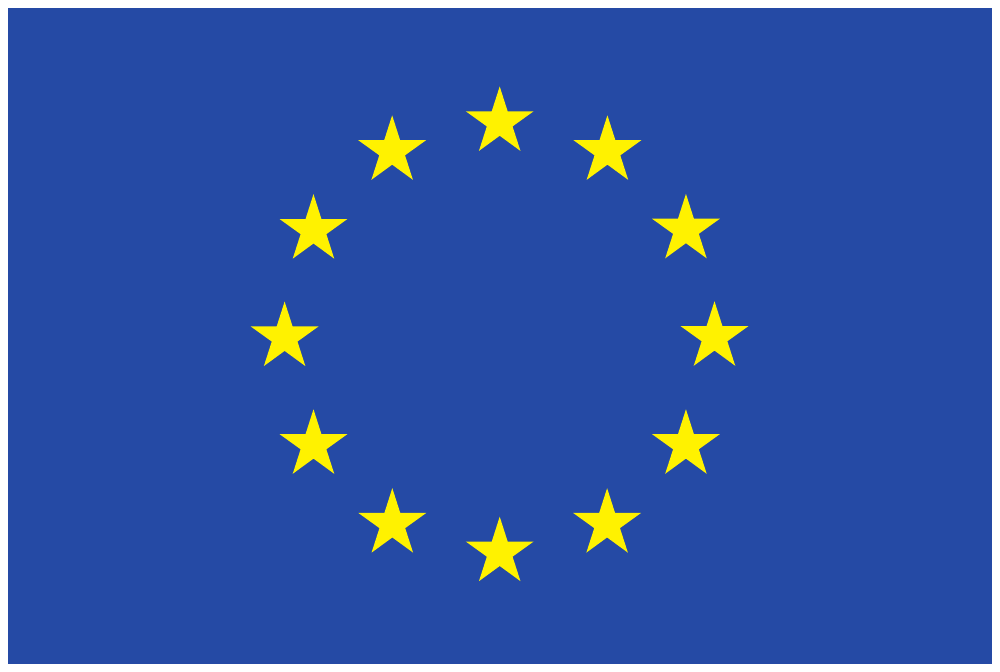}%
\end{textblock}

\begin{abstract}
We study the {\sc{Directed Feedback Vertex Set}} problem parameterized by the treewidth of the input graph.
We prove that unless the Exponential Time Hypothesis fails, the problem cannot be solved in time $2^{o(t\log t)}\cdot n^{\Oh(1)}$ on general directed graphs, where $t$ is the treewidth of the underlying
undirected graph. This is matched by a dynamic programming algorithm with running time $2^{\Oh(t\log t)}\cdot n^{\Oh(1)}$.
On the other hand, we show that if the input digraph is planar, then the running time can be improved to $2^{\Oh(t)}\cdot n^{\Oh(1)}$.
\end{abstract}

\end{titlepage}

\section{Introduction}\label{sect:intro}
In the \ProblemName{Directed Feedback Vertex Set} (\ProblemName{DFVS}) problem we are given a digraph $G$ and the goal is to find a smallest {\em{directed feedback vertex set}} in it, that is,
a subset $X$ of vertices such that $G-X$ is acyclic. The arc-deletion version, \ProblemName{Directed Feedback Arc Set} (\ProblemName{DFAS}), 
differs in that the deletion set $X$ has to consist of edges of $G$ instead of vertices.
The parameterized variants of these problems, where we ask about the existence of a solution of size at most $k$ for a given parameter $k$, are arguably among central problems in the field of parameterized algorithms.
Unfortunately, we are still far from a complete understanding of their complexity.

Establishing the fixed-parameter tractability of \ProblemName{DFVS} was once a major open problem.
It has been resolved by Chen et al.~\cite{ChenLLOR08}, who gave an algorithm for both \ProblemName{DFVS} and \ProblemName{DFAS}\footnote{In general digraphs, \ProblemName{DFVS} and \ProblemName{DFAS}
are well-known to be reducible to each other; see~\cite[Proposition 8.42 and Exercise 8.16]{ParametrizedAlgorithmsBook}. 
These reductions, however, do not preserve planarity of the digraph in question.}
with running time $2^{\Oh(k\log k)}\cdot n^{\Oh(1)}$, obtained by combining iterative compression with a smart application of important separators.
Very recently, Lokshtanov et al.~\cite{LokshtanovRS16} revisited the algorithm of Chen et al.~\cite{ChenLLOR08} and improved the running time to $2^{\Oh(k\log k)}\cdot (n+m)$;
that is, the dependence on the size of the graph is reduced to linear, but the dependence on the parameter $k$ is unchanged.
Whether the running time can be improved to $2^{\Oh(k)}\cdot n^{\Oh(1)}$, or even to $2^{o(k\log k)}\cdot n^{\Oh(1)}$, remains a challenging open problem~\cite{LokshtanovRS16}.
We remark that the question of whether \ProblemName{DFVS} admits a polynomial kernel on general digraphs remains one of the central open problems in the field of kernelization.

A possible reason for why so little progress has been observed on such an important problem, is that the analysis of cut problems in directed graphs is far more complicated than in undirected graphs, and fewer 
basic techniques are available.
For instance, consider the undirected counterpart of the problem, {\sc{Feedback Vertex Set}}, where the goal is to delete at most $k$ vertices from a given undirected graph in order to obtain a forest.
While forests have a very simple combinatorial structure that can be exploited in many ways, acyclic digraphs form a much richer class that cannot be so easily captured.
In particular, undirected graphs admitting a feedback vertex set of size $k$ have treewidth at most $k+1$, and this tree-likeness of positive instances of undirected {\sc{FVS}} makes the problem
amenable to a variety of techniques related to treewidth; other basic techniques like branching and kernelization are also applicable. 
Acyclic digraphs may have arbitrarily large treewidth, whereas directed analogues of treewidth offer almost no algorithmic tools useful for the
design of FPT algorithms.
Therefore, for the study of \ProblemName{DFVS} and other directed cut problems in the parameterized setting, 
we are so far left with important separators and a handful of other more involved techniques; cf.~\cite{ChitnisCHM15,ChitnisHM13,KratschPPW15,KratschW12,PilipczukW16}.

In planar digraphs, the complexity of \ProblemName{DFAS} changes completely. 
As shown by Lucchesi and Younger~\cite{Lucchesi78aminimax}, it is actually polynomial-time solvable (see also a different presentation by Lov\'asz~\cite{lovasz-minimax}).
More precisely, this is a consequence of the proof of the Lucchesi--Younger theorem~\cite{Lucchesi78aminimax}, which states that in planar digraphs, 
the minimum size of a directed feedback arc set is equal to the maximum 
size of a packing of arc-disjoint cycles. The proof is constructive and can be turned into a polynomial-time algorithm that computes a minimum directed feedback arc set together with a maximum cycle packing;
see~\cite{SchrijverBook} for details.

On the other hand, it is easy to see that \ProblemName{DFVS} remains NP-hard on planar digraphs, as there is a simple reduction from {\sc{Vertex Cover}} on planar graphs 
to {\sc{Directed Feedback Vertex Set}} on planar digraphs: just pick an arbitrary ordering of vertices, orient all edges from left to right (giving an acyclic orientation), and replace every edge $uv$ with a directed triangle on $u$, $v$, and a fresh vertex. 
To the best of our knowledge, no algorithm for \ProblemName{DFVS} with running time $2^{o(k\log k)}\cdot n^{\Oh(1)}$ is known even for planar digraphs, which means that so far
we are not able to exploit the planarity constraint in any useful way.

\subparagraph*{Our contribution.} The goal of this paper is to improve our understanding of \ProblemName{DFVS} by studying the parameterization by the treewidth of 
the input directed graph\footnote{Throughout this paper, the treewidth of a directed graph is defined as the treewidth of its underlying undirected graph.}, with a particular focus
 on the planar setting. We first show that a semi-standard dynamic programming approach yields an algorithm with running time $2^{\Oh(t\log t)}\cdot n^{\Oh(1)}$.

\begin{theorem}\label{thm:simple-dp}
There is an algorithm that given a digraph $G$ of treewidth $t$ on $n$ vertices, runs in time $2^{\Oh(t\log t)}\cdot n^{\Oh(1)}$
and determines the minimum size of a directed feedback vertex set and of a directed feedback arc set in $G$.
\end{theorem}

For the proof of Theorem~\ref{thm:simple-dp}, we define the following dynamic programming table (here for \ProblemName{DFVS}). 
For a node $x$ of a tree decomposition of $G$, let $B_x$ be the associated bag and let $G_x$ be the subgraph induced in $G$ by vertices residing in $B_x$ or below $x$ in the decomposition.
Then, for every subset $X$ of $B_x$ and every ordering $\sigma$ of $B_x\setminus X$,
we store the smallest size of a subset $Y$ of $V(G_x)\setminus B_x$ such that $G_x-(X\cup Y)$ is acyclic and admits a topological ordering whose restriction to $B_x\setminus X$ is exactly $\sigma$.
Dynamic programming algorithm for {\sc{DFAS}} is defined similarly.
While we believe that this simple formulation of dynamic programming for \ProblemName{DFVS} and \ProblemName{DFAS}
on a tree decomposition should have been known, we did not find it in the literature and hence we include it for the sake of completeness.

Our next result states then that the running time of the algorithm of Theorem~\ref{thm:simple-dp} is tight under the Exponential Time Hypothesis (ETH) (see the Preliminaries section for definitions).

\begin{theorem}\label{thm:main-lb}
Unless ETH fails, there is no algorithm that determines the minimum size of a directed feedback vertex set or of a directed feedback arc set in
a given digraph in time $2^{o(t\log t)}\cdot n^{\Oh(1)}$, 
where $t$ is the treewidth of the input graph and $n$ is the number of its vertices.
\end{theorem}

The proof of Theorem~\ref{thm:main-lb} uses the approach of Lokshtanov et al.~\cite{LokshtanovMS11} for proving slightly super-exponential lower bounds for the complexity of parameterized problems.
More precisely, we give a parameterized reduction from the \ProblemName{$k\times k$ Hitting Set with thin sets} problem, for which a lower bound excluding running time $2^{o(k\log k)}\cdot n^{\Oh(1)}$ under ETH
was given in~\cite{LokshtanovMS11}.
As an intermediate step, we use problems asking for permutations that satisfy certain constraints; we remark  that somewhat similar constraint satisfaction problems, though with different constraints, were previously studied by Kim and Gon{\c{c}}alves~\cite{KimG13}.

Finally, we move to the setting of planar graphs, where we prove that the running time can be improved to $2^{\Oh(t)}\cdot n^{\Oh(1)}$.

\begin{theorem}\label{thm:main-algo}
There is an algorithm that given a planar digraph $G$ of treewidth $t$ on $n$ vertices, runs in time $2^{\Oh(t)}\cdot n^{\Oh(1)}$
and determines the minimum size of a directed feedback vertex set in $G$.
\end{theorem}

It is well known that the treewidth of a planar graph on $n$ vertices is bounded by $\Oh(\sqrt{n})$; see e.g.~\cite{FominT06}. This yields the following.

\begin{corollary}\label{cor:exact}
There is an algorithm that given a planar digraph $G$ on $n$ vertices, runs in time $2^{\Oh(\sqrt{n})}$
and determines the minimum size of a directed feedback vertex set in $G$.
\end{corollary}

Note that the algorithm of Corollary~\ref{cor:exact} is tight under ETH, due to the aforementioned simple reduction from {\sc{Vertex Cover}} to {\sc{DFVS}} which preserves planarity.
Since {\sc{Vertex Cover}} on planar graphs
cannot be solved in time $2^{o(\sqrt{n})}$ 
under ETH (see~\cite[Theorem~14.6]{ParametrizedAlgorithmsBook}),
the same lower bound carries over to {\sc{DFVS}} on planar digraphs (implying also a tight lower bound of $2^{o(t)} \cdot n^{\Oh(1)}$ for the parameterization by treewidth on planar digraphs).

The proof of Theorem~\ref{thm:main-algo} is perhaps conceptually the most interesting part of this work.
The basic idea is to use {\em{sphere-cut decompositions}} of plane graphs~\cite{DornPBF10,SeymourT94}.
Namely, as observed by Dorn et al.~\cite{DornPBF10}, from the results of Seymour and Thomas~\cite{SeymourT94} it follows that every plane graph admits an optimum-width branch decomposition
that respects the plane embedding in the following sense: each subgraph corresponding to a subtree of the decomposition is embedded into a disk so that the interface of the subgraph---vertices 
adjacent to the remainder of the graph---are embedded on the boundary of the disk. Such a branch decomposition is called a {\em{sphere-cut decomposition}}. Since branchwidth is linearly related to
treewidth, in the proof of Theorem~\ref{thm:main-algo} we may focus on branch decompositions instead of tree decompositions.

As shown by Dorn et al.~\cite{DornPBF10}, the topological properties of sphere-cut decomposition can be exploited algorithmically to bound the number of relevant states in dynamic programming.
This idea is instantiated in the technique of {\em{Catalan structures}} where for some connectivity problems, like {\sc{Hamiltonian Cycle}}, the fact that the solution cannot self-intersect in the plane leads
to an improvement on the number of states from $2^{\Oh(b\log b)}$ to $2^{\Oh(b)}$; here, $b$ is the width of the considered sphere-cut decomposition. However, in the case of {\sc{DFVS}} we cannot use Catalan
structures directly, since we are not building any connected structure whose plane embedding would impose useful constraints.

Our main contribution here is that nevertheless, an improved upper bound on the number of relevant states can be shown, with a conceptually new reasoning.
Consider a directed graph $G$ embedded into a disk $\Delta$ and a subset $T$ of its vertices that are placed on the boundary of $\Delta$.
Let the {\em{connectivity pattern}} induced by $G$ on $T$ be the reachability relation in $G$ restricted to $T^2$: $(s,t)$ are in the connectivity pattern if and only if in $G$ there is a path from $s$ to $t$.
The crucial combinatorial statement (see Theorem~\ref{thm:patterns}) is as follows: 
the number of different connectivity patterns on $T$ that may be induced by different digraphs $G$ embedded in $\Delta$ is bounded by $2^{\Oh(|T|)}$; note
that the naive bound would be $2^{\Oh(|T|^2)}$. This directly provides the sought upper bound on the number of relevant states 
in dynamic programming on a sphere-cut decomposition, leading to the proof of Theorem~\ref{thm:main-algo}.
To prove this statement, we show that every realizable connectivity pattern can be encoded using a constant number of simpler relations, each forming a directed outerplanar graphs on $|T|$ vertices;
the number of different such digraphs is $2^{\Oh(|T|)}$. In the proof that such an encoding is possible we use the result of Gy\'{a}rf\'{a}s that circle graphs are $\chi$-bounded~\cite{Gyarfas85,Gyarfas86}.

\subparagraph*{Organization.} In Section~\ref{sect:prelims} we establish notation and recall known results that will be used throughout the paper.
Section~\ref{sect:connectivityPatterns} concerns the main ingredient of the proof of Theorem~\ref{thm:main-algo}, namely the combinatorial upper bound on the number of different connectivity patterns
induced by disk-embedded directed graphs. We conclude the proof of Theorem~\ref{thm:main-algo} in Section~\ref{sect:dynamic} by giving the dynamic programming algorithm.
Section~\ref{sect:pfthm1} describes the dynamic programming of Theorem~\ref{thm:simple-dp}, Section~\ref{sect:lowerBound} contains the hardness reduction for Theorem~\ref{thm:main-lb},
 whereas in Section~\ref{sect:conclusions} we summarize the results and state some open problems.

\section{Preliminaries}\label{sect:prelims}
\subparagraph*{Notation.} 
For a positive integer $k$, we denote $[k]=\{1,2,\ldots,k\}$. We use standard graph notation, see e.g.~\cite{
	ParametrizedAlgorithmsBook,
	Diestel2010GraphTheoryBook}.
A \emph{clique} of a graph is a set of pairwise adjacent
vertices. The \emph{clique number} of a graph G, denoted by $\omega(G)$, is the maximum number of vertices
in a clique in $G$. The {\em{chromatic number}} of a graph $G$, denoted by $\chi(G)$, is the minimum number of colors needed in a proper coloring of $G$, that is, a coloring of its vertices where
every two adjacent vertices receive different colors.

\subparagraph*{Chords and circle graphs.}
A \emph{chord} is an unordered pair of distinct points on a circle, called \emph{endpoints} of the chord; one may think of it as a straight line segment between its endpoints.
Two chords $\{a,a'\}, \{b,b'\}$ of a circle \emph{cross} if their endpoints are all distinct and $a,b,a',b'$ occur in this order on the circle (clockwise or counter-clockwise). 
Intuitively this corresponds to the straight line segments $aa'$ and $bb'$ intersecting inside the circle.
A \emph{circle graph} is a graph whose vertices correspond to chords of a circle so that two vertices are adjacent if and only if the corresponding chords cross.\footnote{One might consider an alternative
definition where two chords that do not cross but share an endpoint are also connected by an edge. However, this leads to the same class of graphs for the following reason. A set of chords sharing an endpoint
$a$ may be slightly perturbed around $a$ so that they pairwise cross, and they also may be slightly perturbed so that they pairwise do not cross. We choose to use our variant of the definition so that some arguments are smoother.}
A \emph{circle graph with directed chords} is a circle graph in which every chord is directed; that is, it is an ordered pair.
A directed chord with {\em{tail}} $a$ and {\em{head}} $b$ is denoted by $(a,b)$.

Let $T$ be a finite set of points on a circle and let $R\subseteq T^2$ be a set of chords (directed or undirected). 
A \emph{crossing} is a pair of crossing chords in $R$.
The circle graph {\em{induced}} by $R$ is the one with $R$ as the vertex set where two chords from $R$ are adjacent if they cross.

As introduced by Gy\'{a}rf\'{a}s~\cite{gyarfas1987problems}, 
a class $\mathcal{C}$ of graphs closed under induced subgraphs is \emph{$\chi$-bounded} if there exists a function $ f \colon \mathbb{N} \rightarrow \mathbb{N}$ such that for every graph $G \in \mathcal{C}$ we have $\chi(G)\leq f(\omega(G))$.
Gy\'{a}rf\'{a}s \cite{Gyarfas85,Gyarfas86} proved the following.

\begin{theorem}[\cite{Gyarfas85,Gyarfas86}]\label{Gyarfas}
The class of circle graphs is $\chi$-bounded.
\end{theorem}

\subparagraph*{Exponential Time Hypothesis.}
The Exponential Time Hypothesis (ETH) states that for some constant $c>0$, there is no algorithm for {\sc{3SAT}}
with running time $\Oh(2^{cn})$, where $n$ is the number of variables of the input formula. Since its formulation by Impagliazzo, Paturi, and Zane~\cite{ImpagliazzoPZ01},
ETH has served as a basic assumption for countless lower bounds on the complexity of computational problems.
We refer to~\cite[Chapter 14]{ParametrizedAlgorithmsBook} for a comprehensive overview of applications in parameterized complexity.

In this work we do not use ETH directly, but we rely on results of Lokshtanov et al.~\cite{LokshtanovMS11} 
who introduced a methodology for refuting ``slightly super-exponential'' running time for parameterized problems, under ETH.
In particular, the hardness reduction proving Theorem~\ref{thm:main-lb} starts with the \ProblemName{$k\times k$ Hitting Set with thin sets} problem, for which a lower bound is given in~\cite{LokshtanovMS11}.
Relevant definitions and results are recalled in Section~\ref{sect:lowerBound}.

\section{Connectivity Patterns}\label{sect:connectivityPatterns}
\newcommand{\clos}[1]{\mathsf{gen}(#1)}

In this section we present the main combinatorial result leading to the proof of Theorem~\ref{thm:main-algo}, which is a reduction of the number of relevant dynamic programming states in the planar setting.
As we mentioned, this is done by bounding the number of ``connectivity patterns'' that can be induced by directed graphs embedded in a disk.
We now formalize this idea.

Suppose $T$ is a finite set. A \emph{connectivity pattern} on $T$ is any quasi-order on $T$, that is, a reflexive and transitive relation $P\subseteq T^2$.
For a directed graph $G$ and a vertex subset $T\subseteq V(G)$, we define the connectivity pattern \emph{induced by $G$ on $T$} to be the reachability relation on $T$ in $G$:  
$(s,t)$ is in the relation iff there is a path in $G$ from $s$ to $t$.

The main goal of this section is to prove a result that roughly states the following: for a directed graph $G$ drawn in a closed disk, with $T$ be the vertices lying the boundary of the disk, 
there are only $2^{\Oh(|T|)}$ different possibilities for the connectivity pattern that $G$ may induce. See Theorem~\ref{thm:patterns} for a formal statement.
As mentioned in the introduction, this result will be our main tool for limiting the number of relevant states in dynamic programming for {\sc{Directed Feedback Vertex Set}} on planar graphs.
Note that in general directed graphs, the number of different connectivity patterns induced on a vertex subset $T$ may be as large as $2^{\Theta(|T|^2)}$.
For instance, any subset of pairs with tail in the first half of $T$ and head in the second half already gives that many possibilities.

The idea of the proof is that such connectivity patterns induced by directed planar graphs embedded in a disk can be generated from simpler relations, 
which contain enough pairs to infer all the other ones from planarity.
This is formalized in the following definition.

\begin{definition}
For a set $T$ of points on a circle and a relation $R\subseteq T^2$,
define the connectivity pattern on $T$ \emph{generated by $R$}, denoted $\clos{R}$, as follows:
a pair $(s,t)\in T^2$ is included in $\clos{R}$ if and only if
for each partition of the circle into two disjoint arcs $X_s,X_t$ such that $s\in X_s$ and $t\in X_t$,
there exist $s'\in X_s$ and $t'\in X_t$ which satisfy $(s',t')\in R$.
\end{definition} 

In the above definition, as well as throughout this whole section, arcs on a circle may be open or closed from either side, unless explicitly stated.

It is easy to check that $R \subseteq \clos{R}$ and $\clos{R}$ is indeed reflexive and transitive, for any $R\subseteq T^2$.
Hence $\clos{R}$ also contains the reflexive transitive closure of $R$, but it may be much larger still.
Furthermore, one can observe that $\clos{\clos{R}} = \clos{R}$, but we will not use this property.
We now show that a connectivity pattern induced by a graph is generated by itself; the goal will be then to find simpler relations generating this pattern.

\begin{lemma}\label{pgenerp}
Let $G$ be a planar digraph drawn in a disk $\Delta$, $T$ be a subset of vertices drawn on the boundary of $\Delta$, and $P$ be the connectivity pattern on $T$ induced by $G$. Then $\clos{P} = P$.
\end{lemma}
\begin{proof}
Let $C$ be the boundary of $\Delta$; we may assume that $C$ is a circle.
Clearly $P \subseteq \clos{P}$.
Now assume that $(s,t) \in \clos{P}$, that is, for each partition of $C$ into two disjoint arcs $X_s,X_t$ such that $s\in X_s$ and $t\in X_t$, there exist
$s'\in X_s$ and $t'\in X_t$ which satisfy
$(s',t')\in P$. We will show that $(s,t)\in P$. 

Assume the contrary, that is, $(s,t)\notin P$. Define 
$T_s=\{r\in T\colon (s,r)\in P\}$, see Figure~\ref{fig:P-generates-P}.
Let $X_t$ be the largest arc on $C$ that contains $t$ and is disjoint from $T_s$; this is well-defined since $t\notin T_s$ and $s \in T_s$. 
Define $X_s=C\setminus X_t$, thus $(X_s,X_t)$ is a partition of $C$ into two disjoint arcs. 
Since $s\in T_s$, we have $s\notin X_t$ and thus $s\in X_s$. From our assumption that $(s,t) \in \clos{P}$, there exist
$s'\in X_s$ and $t'\in X_t$ that satisfy $(s',t')\in P$.

We have two cases: either $s'\in T_s$ or $s'\notin T_s$. 
If $s'\in T_s$, then $(s,s')\in P$ and consequently $(s,t')\in P$, since $P$ is transitive due to being the reachability relation induced by $G$. 
But then $t'\in T_s$ and hence $t'\notin X_t$, a contradiction. Now assume 
$s'\notin T_s$; in particular $s'\neq s$. Let us move along the circle from $s$ to $t$ such that on the way we meet the point $s'$. Because the arc
$X_t$ was chosen to be the largest possible, 
between $s'$ and $t$ we meet a point $r\in T_s$. 
The arc $X_t$ is connected, so between $s$ and $r$ we cannot meet any point from the set $X_t$, in particular $t'$. 
That is, $s,s',r,t'$ appear in this order on the circle (either clockwise or counterclockwise).
Since $r \in T_s$, we have $(s,r) \in P$ and $(s',t') \in P$. Therefore, in $G$ there are directed paths from $s$ to $r$ and from $s'$ to $t'$. These two paths must intersect since they are drawn in a disk, which
yields a path in $G$ from $s$ to $t'$. We conclude that $t'\in T_s$ and hence $t'\notin X_t$, a contradiction. 
\end{proof}

\begin{figure}[b]
	\centering
	\parbox{0.4\linewidth}{
		\begin{tikzpicture}[->,>=stealth,thick]
	\def\centerarc[#1](#2)(#3:#4:#5){ \draw[#1] ([shift=(#3:#5)]#2) arc (#3:#4:#5); }

	\tikzstyle{Ts} = [black!30!green!70!blue,circle,fill,minimum size=0.5]
	\tikzstyle{t} = [black,rectangle,draw,fill=blue!10!white,minimum size=1]
	\def\rr{1.7}
	\draw[black!55!white,thick] (0,0) circle [radius=\rr];

	\node[Ts,label=below:$s$] (s) at (90:\rr) {};
	\node[Ts,label=-10:$r$] (r0) at (190:\rr) {};
	\node[Ts] (r1) at (30:\rr) {};
	\node[Ts] (r2) at (-25:\rr) {};
	\node[Ts,fill=white!40!gray,label=95:$t$] (t) at (-90:\rr) {};
	\node[Ts,fill=white!40!gray,label=15:$s'$] (sp) at (155:\rr) {};
	\node[Ts,fill=white!40!gray,label=90:$t'$] (tp) at (-65:\rr) {};
	\draw (s)--(r0);
	\draw (s)--(r1);
	\draw (s)--(r2);
	\draw (sp)--(tp);

	\centerarc[black!55!blue,thick,*-*](0,0)(-28:193:\rr+0.5);
	\centerarc[black!55!red,thick,)-(](0,0)(194:360-29:\rr+0.5);
	\node at (40:\rr+0.85) {$X_s$};
	\node at (-40:\rr+0.85) {$X_t$};
\end{tikzpicture}
		\caption{Proof of Lemma~\ref{pgenerp}: the induced pattern $P$ shown as arrows, points in $T_s$ depicted in green.}%
		\label{fig:P-generates-P}%
	}
	\qquad
	\begin{minipage}{0.4\linewidth}
		\begin{tikzpicture}[->,>=stealth,thick]
	\def\centerarc[#1](#2)(#3:#4:#5){ \draw[#1] ([shift=(#3:#5)]#2) arc (#3:#4:#5); }

	\tikzstyle{v} = [gray!85!blue,circle,fill,minimum size=0.5]
	\def\rr{1.7}

	\draw[black!55!white,thick] (0,0) circle [radius=\rr];

	\node[v,label=right:$a$] (a) at (90+45:\rr) {};
	\node[v,label=left:$b$] (b) at (45:\rr) {};
	\node[v,label=left:$a'$] (ap) at (-45:\rr) {};
	\node[v,label=right:$b'$] (bp) at (-90-45:\rr) {};
	\draw (a)--(ap);
	\draw (b)--(bp);
	\draw[gray,dashed] (a)--(bp);

	\centerarc[black!55!blue,thick,)-*](0,0)(-28:193:\rr+0.5);
	\centerarc[black!55!red,thick,)-*](0,0)(194:360-29:\rr+0.5);
	\node at (30:\rr+0.85) {$X_s$};
	\node at (-50:\rr+0.85) {$X_t$};
	\node at (193:\rr+0.85) {$u$};
	\node at (-28:\rr+0.85) {$v$};
\end{tikzpicture}
		\caption{Proof of Lemma~\ref{addedge}.
		}
		\label{fig:addEdge}%
	\end{minipage}
\end{figure}

The next lemma shows that generated connectivity patterns are closed under adding directed chords $(a,b')$ whenever $(a,a')$ and $(b,b')$ cross. 
This operation (and its inverse) is in fact the only one we will use to simplify the generating relation.


\begin{lemma}\label{addedge}
Let $T$ be a finite set of points on a circle and let $R\subseteq T^2$.
Let $a,b,a',b'\in T$ be distinct points that appear in this order on the circle, such that $(a,a')\in R$ and $(b,b')\in R$.
Let $R'=R\cup\{(a,b')\}$. 
Then $\clos{R} = \clos{R'}$.
\end{lemma}
\begin{proof}
It is enough to prove that for each partition of the circle into two disjoint arcs $X_s,X_t$,
the following two conditions are equivalent:
\begin{enumerate}[nosep]
\item[(1)] There exist
$s'\in X_s$ and  $t'\in X_t$ which satisfy
$(s',t')\in R$. 
\item[(2)] There exist
$s'\in X_s$ and  $t'\in X_t$ which satisfy
$(s',t')\in R'$.
\end{enumerate}
Of course (1) implies (2). Now assume (2). 
If $(s',t')\in R$ the proof is finished, so suppose $(s',t')=(a,b')$. Let $u,v$ be the ends of the arc $X_s$, see Figure~\ref{fig:addEdge}.
We may assume without loss of generality that $a,b,a',b'$ occur clockwise on the circle and are different from $u,v$; the latter is achieved by moving $u,v$ slightly to points not belonging to $T$.
Let $C_{a,b}$ be the arc of the circle from $a$ (inclusive) to $b$ (exclusive), going clockwise, and define $C_{b,a'},$ $C_{a',b'},$ $C_{b',a}$ analogously; these four arcs form a partition of the circle.
Since $a\in X_s$ and  $b'\in X_t$, we may assume that $u\in C_{b',a}$ and $v\notin C_{b',a}$. 
If $v\in C_{a,b}$ or $v\in C_{b,a'}$, then $a\in X_s$ and  $a'\in X_t$ satisfy $(a,a')\in R$. 
Otherwise, if $v\in C_{a',b'}$, then $b\in X_s$ and  $b'\in X_t$ satisfy $(b,b')\in R$.  In both cases, (1) holds.
\end{proof}

Next, we prove that the generating relation can be simplified as long as it contains $4$ pairwise crossing chords in the right order. 
The lemma after that shows how to obtain such chords from any set of $7$ pairwise crossing chords.

\begin{figure}[b!]
	\centering
	\begin{tikzpicture}[->,>=stealth,yscale=0.63,xscale=0.63]
	\def\centerarc[#1](#2)(#3:#4:#5){ \draw[#1] ([shift=(#3:#5)]#2) arc (#3:#4:#5); }

	\tikzstyle{vv} = [gray!85!yellow,circle,fill,minimum size=1,inner sep=2pt]
	\tikzstyle{rred} = [black!50!red]
	\def\rr{1.6}

	\newcommand\mysubdrawing{
		\draw[black!55!white,thick] (0,0) circle [radius=\rr];
		\node[vv,label=above:$a$] (a) at (2+72+72:\rr) {};
		\node[vv,label=above:$b$] (b) at (2+72+36:\rr) {};
		\node[vv,label=above:$c$] (c) at (2+72:\rr) {};
		\node[vv,label=above:$d$] (d) at (2+36:\rr) {};
		\node[vv,label=below:$u$] (u) at (2-72-72:\rr) {};
		\node[vv,label=below:$z$] (z) at (2-72-36:\rr) {};
		\node[vv,label=below:$y$] (y) at (2-72:\rr) {};
		\node[vv,label=below:$x$] (x) at (2-36:\rr) {};
	}

	\begin{scope}[shift={(-4,0)}]
		\mysubdrawing
		\draw (a)--(x);
		\draw (d)--(u);
		\draw (b)--(z);
		\draw (c)--(y);
		\node at (-180:1) {$R'$};
	\end{scope}
	\begin{scope}[shift={(0,0)}]
		\mysubdrawing
		\draw[very thick, rred] (a)--(x);
		\draw (d)--(u);
		\draw[very thick, rred] (b)--(z);
		\draw (c)--(y);
		\draw[dashed] (b)--(x);
	\end{scope}
	\begin{scope}[shift={(4,0)}]
		\mysubdrawing
		\draw (a)--(x);
		\draw[very thick, rred] (d)--(u);
		\draw (b)--(z);
		\draw[very thick, rred] (c)--(y);
		\draw (b)--(x);
		\draw[dashed] (c)--(u);
	\end{scope}
	\begin{scope}[shift={(8,0)}]
		\mysubdrawing
		\draw (a)--(x);
		\draw (d)--(u);
		\draw (b)--(z);
		\draw[very thick, rred] (c)--(y);
		\draw[very thick, rred] (b)--(x);
		\draw (c)--(u);
		\draw[dashed] (b)--(y);
	\end{scope}
	\begin{scope}[shift={(12,0)}]
		\mysubdrawing
		\draw (a)--(x);
		\draw (d)--(u);
		\draw[very thick, rred] (b)--(z);
		\draw (c)--(y);
		\draw (b)--(x);
		\draw[very thick, rred] (c)--(u);
		\draw (b)--(y);
		\draw[dashed] (c)--(z);
	\end{scope}

\begin{scope}[shift={(0,-0.5)}]
	\begin{scope}[shift={(-4,5)}]
		\mysubdrawing
		\draw (a)--(x);
		\draw (b)--(y);
		\draw (c)--(z);
		\draw (d)--(u);
		\node at (-180:1) {$R$};
	\end{scope}
	\begin{scope}[shift={(0,5)}]
		\mysubdrawing
		\draw (a)--(x);
		\draw[very thick,rred] (b)--(y);
		\draw[very thick,rred] (c)--(z);
		\draw (d)--(u);
		\draw[dashed] (b)--(z);
	\end{scope}
	\begin{scope}[shift={(4,5)}]
		\mysubdrawing
		\draw (a)--(x);
		\draw[very thick,rred] (b)--(y);
		\draw[very thick,rred] (c)--(z);
		\draw (d)--(u);
		\draw (b)--(z);
		\draw[dashed] (c)--(y);
	\end{scope}
	\begin{scope}[shift={(8,5)}]
		\mysubdrawing
		\draw[very thick,rred] (a)--(x);
		\draw[very thick,rred] (b)--(y);
		\draw (c)--(z);
		\draw (d)--(u);
		\draw (b)--(z);
		\draw (c)--(y);
		\draw[dashed] (b)--(x);
	\end{scope}
	\begin{scope}[shift={(12,5)}]
		\mysubdrawing
		\draw (a)--(x);
		\draw (b)--(y);
		\draw[very thick,rred] (c)--(z);
		\draw[very thick,rred] (d)--(u);
		\draw (b)--(z);
		\draw (c)--(y);
		\draw (b)--(x);
		\draw[dashed] (c)--(u);
	\end{scope}
\end{scope}
\end{tikzpicture}
	\caption{Proof of Lemma~\ref{change}: starting from the left $R$ or $R'$, we can add the dashed chord without changing the generated set, due to the thick red ones. Eventually we reach the same $R''$ on the right.}%
	\label{fig:exchange}%
\end{figure}
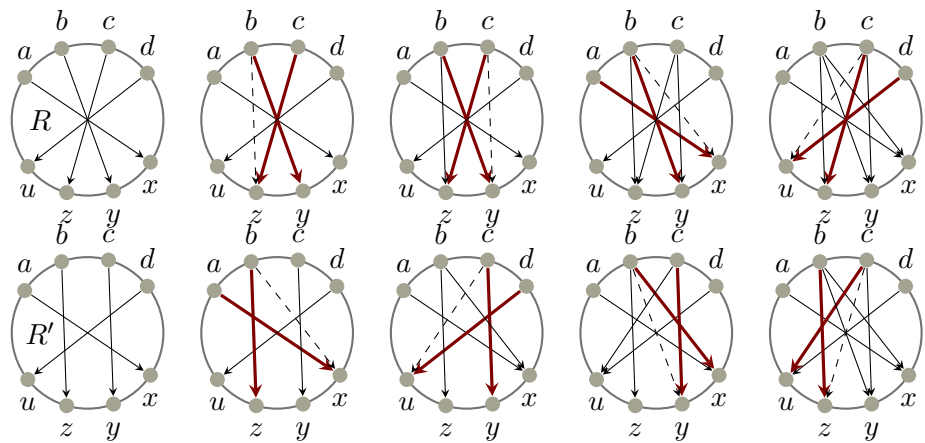

\begin{lemma}\label{change}
Let $T$ be a finite set of points on a circle and let $R\subseteq T^2$.
Let $a,b,c,d,x,y,z,u\in T$ be pairwise different points that appear in this order on the circle (clockwise or counterclockwise), such that $(a,x),(b,y),(c,z),(d,u)\in R$.
Define 
$$R'=(R\setminus \{(b,y),(c,z)\})\cup\{(b,z),(c,y)\}.$$ 
Then $\clos{R'} = \clos{R}$ and the number of crossings in $R'$ is strictly smaller than in~$R$.
\end{lemma}
\begin{proof}
Let $R'':=R\cup \{(b,x),(b,z),(c,y),(c,u)\}$; then $R''=R'\cup \{(b,x),(b,y),(c,z),(c,u)\}$.
By consecutively applying Lemma~\ref{addedge} to $R$ (see Figure~\ref{fig:exchange}) and quadruples
$$(b,c,y,z),\quad (c,b,z,y),\quad (b,a,y,x),\quad (c,d,z,u),$$
we infer that $\clos{R}=\clos{R''}$.
Similarly, by applying Lemma~\ref{addedge} to $R'$ and quadruples
$$(b,a,z,x),\quad (c,d,y,u),\quad (b,c,x,y),\quad (c,b,u,z),$$
we infer that $\clos{R'}=\clos{R''}$. Thus $\clos{R}=\clos{R''}=\clos{R'}$, proving the first claim.

We are left with proving that $R'$ has strictly fewer crossings than $R$.
Since $R'=(R\setminus \{(b,y),(c,z)\})\cup\{(b,z),(c,y)\}$, chords $(b,y)$ and $(c,z)$ cross, and chords $(b,z)$ and $(c,y)$ do not cross, it suffices to prove the following:
every chord $(s,t)$ crosses at most as many chords in the set $\{(b,z),(c,y)\}$ as in the set $\{(b,y),(c,z)\}$.
This assertion can be directly checked by a straightforward case distinction over the positions of $s,t$ with respect to $b,c,y,z$ on the circle.
\end{proof}

\begin{lemma}\label{clique}
Suppose $H$ is a circle graph with directed chords such that $\omega(H)\geq 7$.
Then there are pairwise different points $a,b,c,d,x,y,z,u$ that appear in clockwise order on the circle such that $(a,x),(b,y),(c,z),(d,u)$ are pairwise crossing chords of $H$.
\end{lemma}
\begin{proof}
We first show that there is a clique $V_1$ of four chords and a partition of the circle into two arcs $X_1, X_2$ such that each of the four chords has tail in $X_1$ and head in $X_2$.
Let $C$ denote the circle and $V_2$ be a clique in $H$ with 
$|V_2|=7.$ Choose any chord $(a,b)\in V_2$. After removing points $a$ and $b$, the circle $C$ breaks into 
two disjoint open arcs $C_1,C_2$ with $a,b$ as endpoints.
Every other chord of $(x,y)\in V_2$ crosses the chord $(a,b)$, hence exactly one of the endpoints of $(x,y)$ belongs
to the arc $C_1$, and the other belongs to the arc $C_2$. By the pigeonhole principle, some
three of the six chords in $V_2 \setminus \{(a,b)\}$ have their tail in the same arc $C_i$, for some $i \in \{1,2\}$. Define $V_1$ to be the set of these three chords together with $(a,b)$.
Then the (open-closed) arcs $X_1 := C_i \cup \{a\}$ and $X_2 := C_{3-i} \cup \{b\}$ have the property that each chord in $V_1$ has tail in $X_1$ and head in $X_2$.

Let $a,b,c,d\in X_1$ be the tails of chords in $V_1$, in the clockwise order on $C$. Similarly let $x,y,z,u\in X_2$ be the heads of chords in $V_1$, in the clockwise order on $C$. 
We claim that $V_1=\{(a,x),(b,y),(c,z),(d,u)\}$, which will conclude the proof.

Suppose $(a,\zeta)\in V_1$ for some $\zeta \in \{y,z,u\}$. Then
$(\eta,x)\in V_1$ for some 
$\eta \in \{b,c,d\}$. 
Since $V_1$ is a clique, the chords $(\eta,x)$ and $(a,\zeta)$ have to cross. But this contradicts that $a,\eta,x,\zeta$ appear in this clockwise order on the circle. Therefore $(a,x)\in V_1$, and symmetrically
$(d,u)\in V_1$. Now 
either $(b,y),(c,z)\in V_1$ or $(b,z),(c,y)\in V_1$.
The latter is impossible, since  $(b,z),(c,y)$ would not cross, hence $(b,y),
(c,z)\in V_1$. 
\end{proof}

Lemmas~\ref{change} and~\ref{clique} allow us to conclude that any generating relation can be iteratively simplified until it contains no set of $7$ pairwise crossing chords.

\begin{lemma}\label{less6}
Let $G$ be a planar graph drawn in a disk $\Delta$, let $T$ be a subset of vertices of $G$ drawn on the boundary of $\Delta$, and let $P\subseteq T^2$ be the connectivity pattern on $T$ induced by $G$. 
Then there exists a relation $R\subseteq T^2$ such that
$\clos{R} = P$ and the circle graph induced by $R$ has clique number at most $6$.
\end{lemma}
\begin{proof}
By Lemma \ref{pgenerp} there exists a relation $R\subseteq T^2$ (namely $R=P$) such that $\clos{R}=P$.
Choose $R$ such that $\clos{R} = P$ and the number of crossings in $R$ is as small as possible.
Without loss of generality assume that $R$ does not contain pairs of the form $(s,s)$ for $s\in T$, as such pairs may be removed without changing the generated relation; 
thus $R$ is a set of directed chords with endpoints in $T$.
Let $\omega$ be the clique number of the circle graph induced by $R$.
If $\omega\leq 6$ we are done, so suppose $\omega\geq 7$.  By Lemma~\ref{clique}, there are pairwise different points $a,b,c,d,x,y,z,u$ that appear in clockwise order on the circle such that 
$(a,x),(b,y),(c,z),(d,u)\in R.$ 
Define 
$$
R'=(R\setminus \{(b,y),(c,z)\})\cup\{(b,z),(c,y)\}.
$$ 
By Lemma~\ref{change}, $\clos{R'}=\clos{R} = P$ and $R'$ has fewer crossings than $R$, a contradiction.
\end{proof}

Having obtained a generating relation with no large set of pairwise crossing chords, we will later partition it into a small number of sets of pairwise non-crossing chords 
using the $\chi$-boundedness of circle graphs (Theorem \ref{Gyarfas}). First, however, we bound the number of such non-crossing sets as follows.

\begin{lemma}\label{disjoint}
Let $T$ be a finite set of points on a circle.
Then every set of pairwise non-crossing chords with endpoints in $T$ has at most $2|T|-3$ chords, and
there are $2^{\Oh(|T|)}$ different such sets.
\end{lemma}
\begin{proof}
Let $|T|=n$. 
Observe that any set of pairwise non-crossing chords with endpoints in $T$ corresponds to an outerplanar graph with $T$ as vertices and chords as edges.
It is well known that the number of edges in an outerplanar graphs on $n$ vertices is at most $2n-3$, and the number of maximal outerplanar graphs (polygonal triangulations) on $n$ vertices 
is the Catalan number $C_{n-2}$.
Any outerplanar graph on $n$ vertices can be obtained as a subgraph of a maximal one, hence their number is bounded by $C_{n-2} \cdot 2^{2n-3} \leq 4^{n} \cdot 2^{2n-3} = 2^{\Oh(n)}$.
\end{proof}

We remark that in the proof above, the exact number of possibilities can also be characterized in terms of little Schr\"{o}der numbers $s_{n-2}$ (also known as Schr\"{o}der-Hipparchus numbers), 
which count the number of dissections of a polygon with $n$ sides by non-crossing diagonals~\cite{oeisA001003,stanley1997hipparchus}.
Since there are $2^n$ ways to choose a subset of non-diagonals (sides of the polygon), the exact number is $s_{n-2} \cdot 2^n$.


We are now ready to prove the main theorem of this section.

\newcommand{\chimax}{\chi_{\mathrm{max}}}

\begin{theorem}\label{thm:patterns}
Let $T$ be a set of $n$ points on the boundary of a closed disk $\Delta$.
There exists a family $\mathcal{R}$ of relations $R\subseteq T^2$ such that $|\mathcal{R}| = 2^{\Oh(n)}$ and the following property is satisfied.
For every planar digraph $G$ drawn in $\Delta$ such that $T\subseteq V(G)$, 
the connectivity pattern induced by $G$ on $T$ is generated by some relation in $\mathcal{R}$.
\end{theorem}
\begin{proof}
Denote by $\mathcal{R}$ the family of all sets of directed chords $R\subseteq T^2$ such that the circle graph induced by $R$ has clique number at most $6$.
By Lemma~\ref{less6} this family satisfies the claimed property and it remains to bound its size.

By $\chi$-boundedness of circle graphs (Theorem~\ref{Gyarfas}), there exists a number $\chimax$ such that for $R\in \mathcal{R}$, 
the chromatic number of the circle graph induced by $R$ is at most $\chimax$.
The chords of any circle graph induced by some $R\in\mathcal{R}$ can thus be partitioned into $\chimax$ sets (possibly empty) such that no two chords in the same set cross.
By Lemma~\ref{disjoint}, the number of possibilities to choose such a set of undirected, pairwise non-crossing chords is $2^{\Oh(n)}$,
and any such set contains at most $2n-3$ chords. Hence there are at most $2^{2n-3}$ possibilities to orient these chords.
We conclude that indeed $|\mathcal{R}| \leq (2^{\Oh(n)} \cdot 2^{2n-3})^{\chimax} = 2^{\Oh(n)}$.
\end{proof}

\section{Dynamic programming on planar digraphs}\label{sect:dynamic}
\newcommand{\med}{\mathsf{med}}
\newcommand{\conn}{\mathsf{conn}}
\newcommand{\Val}{\mathsf{Val}}
\newcommand{\bw}{\text{bw}}

In this section we prove Theorem~\ref{thm:main-algo}, that is, we describe the algorithm for \ProblemName{DFVS} on planar directed graphs running in time $2^{\Oh(t)}\cdot n^{\Oh(1)}$.

\subparagraph*{Branch decompositions.}
We use standard terminology for rooted trees: children, parent, etc. Moreover,
if $x$ is a node of a rooted tree, $e$ is the edge connecting $x$ to its parent, whereas $e'$ is an edge connecting $x$ to one of its children, then $e'$ is called a {\em{child}} of $e$.

For a graph $G$, a \emph{(rooted) branch decomposition} of $G$ is a pair $(T,\eta)$ where $T$ is a tree with a root 
$r$ of degree 1 and all internal vertices of degree 3, and $\eta$ is a bijection from the edge set of $G$ to the set of leaves of $T$ (the root is not considered a leaf).
For an edge $e \in E(T)$ of the decomposition, removing $e$ from $T$ splits the tree into two subtrees; let $T_{\uparrow}(e)$ be the subtree that contains the root, and let $T_{\downarrow}(e)$ be the other subtree.
This naturally gives rise to a partition $E_{\uparrow}(e), E_{\downarrow}(e)$ of the edge set of $G$, where $E_{\uparrow}(e)$, resp. $E_{\downarrow}(e)$, comprises edges mapped by $\eta$ to leaves lying in 
$T_{\uparrow}(e)$, resp. $T_{\downarrow}(e)$.
Note that both sides of this partition are non-empty unless $e$ is the edge incident to the root.
We define $\med(e) \subseteq V(G)$ as the set of vertices incident to both an edge in $E_{\uparrow}(e)$ and $E_{\downarrow}(e)$.
Note that $\med(e)$ is non-empty unless $e$ is the edge incident to the root.
Let $G(e)$ be the subgraph of $G$ containing edges of $E_{\downarrow}(e)$ and their endpoints.
The \emph{width} of the decomposition $(T,\eta)$ is $\max_{e \in E(T)} |\med(e)|$.
The \emph{branchwidth} of $G$ is the minimum width of a branch decomposition of $G$.
It is well known that the branchwidth of a graph is at least its treewidth minus $1$ and at most $\frac{3}{2}$ times its treewidth~\cite{RobertsonS91}.
Hence, it suffices to prove Theorem~\ref{thm:main-algo} with the treewidth of the input graph replaced with its branchwidth.

\subparagraph*{Sphere-cut decompositions.}Let $G$ be a graph embedded on a sphere (a plane graph).
A {\em{noose}} of $G$ is a Jordan curve (i.e., closed, non self-crossing) on the sphere that intersects $G$ only in vertices and visits each face at most once.
A \emph{sphere-cut decomposition} (\emph{sc-decomposition}, for short) of $G$ is a branch decomposition $(T,\eta)$ of $G$ such that for each $e \in E(T)$, 
there is a closed disk $\Delta(e)$ on the sphere such that $\Delta(e) \cap G = G(e)$ and $\partial \Delta(e)$---the boundary of the disk $\Delta(e)$---is a noose of $G$ with $\partial \Delta(e) \cap G = \med(e)$.
Moreover, for an edge $e$ of $T$ with children $e_1, e_2$, the interiors of disks $\Delta(e_1)$ and $\Delta(e_2)$ are disjoint and $\Delta(e_1) \cup \Delta(e_2) = \Delta(e)$.

We use the following results on the existence of sc-decompositions.
It follows from the results of Seymour and Thomas~\cite{SeymourT94} that every sphere-embedded graph that is connected and bridgeless admits a sphere-cut decomposition of width equal to its branchwidth.
This was observed and used by Dorn et al.~\cite{DornPBF10}; see also Marx and Pilipczuk~\cite{MarxP15} for a slightly corrected explanation.

All the works mentioned above use non-rooted decompositions, however these can be rooted by subdividing an arbitrary edge and creating a root as a new node attached to the middle vertex of the subdivided edge.
The definition of branchwidth in particular is unchanged.
We also remark that the definition of a sphere-cut decomposition in the works above relies only on the existence of nooses that meet the embedded graph $G$ at vertex subsets $\med(e)$, 
for edges $e$ of the decomposition, and does not assert the properties of disks $\Delta(e)$ that we imposed in our definition.
However, the satisfaction of these properties is implicit in~\cite{FominT06,MarxP15}; this follows from the fact that nooses may be treated as simple cycles in the radial graph~\cite{MarxP15}.

Moreover, it is known that an optimum-width sc-decomposition of a planar graph can be computed in polynomial time.
The original running time of $\Oh(n^4)$ given by Seymour and Thomas~\cite{SeymourT94} has been improved to $\Oh(n^3)$ by Gu and Tamaki~\cite{GuT08}. 
See also Bian et al.~\cite{BianGZ16} for experimental results.
We summarize all these results in the following theorem.

\begin{theorem}[\cite{DornPBF10,GuT08,MarxP15,SeymourT94}]\label{thm:planarScDecomposition}
Let $G$ be a sphere-embedded, connected, bridgeless graph of branchwidth $b$ on $n$ vertices. Then an sc-decomposition of $G$ of width $b$ exists and can be found in $\Oh(n^3)$ time.
\end{theorem}

\subparagraph*{Connectivity patterns and joins.}
Suppose $M$ is a set of vertices on the boundary of a closed disk $\Delta$. Let $\conn(M)$ be the family of all connectivity patterns that are induced by digraphs $H$ embedded in $\Delta$
with $M \subseteq \partial \Delta \cap H$. By Theorem~\ref{thm:patterns}, $\conn(M)$ has size at most $2^{\Oh(|M|)}$.

Fix a sphere-embedded directed graph $G$ and any its sc-decomposition.
For an edge $e$ of the sc-decomposition
and any set $M\subseteq \med(e)$,
when we write $\conn(M)$, we mean $\conn(M)$ as defined above for the vertices $\med(e)$ embedded on the boundary of the disk $\Delta(e)$.

For an edge $e$ of the sc-decomposition with children $e_1,e_2$, a vertex subset $X \subseteq V(G)$, and connectivity patterns $P_i \in \conn(\med(e_i) \setminus X)$ for $i=1,2$, 
we define the \emph{join} of $P_1$ and $P_2$ as the connectivity pattern in $\conn(\med(e) \setminus X)$ obtained as follows.
For $i=1,2$ define the directed graph $Q_i=(\med(e_i) \setminus X, P_i)$ (i.e. vertices are $\med(e_i) \setminus X$ and edges are pairs in $P_i$),
take the union of $Q_1$ and $Q_2$ (note that $Q_1$ and $Q_2$ share vertices from $(\med(e_1)\cap \med(e_2))\setminus X$),
and define the join of $P_1$ and $P_2$ to be the connectivity pattern induced by this union on $\med(e) \setminus X$, provided it is acyclic.
If the union contains directed cycles (other than loops), we say that $P_1$ and $P_2$ have no join.
By the definition of a branch decomposition we have that $\med(e) \subseteq \med(e_1) \cup \med(e_2)$, so this is notion is well defined.

The crucial property of joins, which follows directly from the properties of disks defined in sc-decompositions, is that for a subset $X$ of vertices or edges of $G$ such that $G(e_i) - X$ is acyclic (for $i=1,2$), the
following assertion holds:
The connectivity patterns induced by $G(e_i) - X$ on $\med(e_i) \setminus X$ have a join if and only if $G(e) -X$ is acyclic, and
in this case the join is equal to the connectivity pattern induced by $G(e) - X$ on $\med(e) \setminus X$.

\subparagraph*{Dynamic programming over sc-decompositions.}
Consider an instance $G$ of \ProblemName{DFVS}, cast as an optimization problem.
We can assume without loss of generality that $G$ is weakly connected and has no bridges, since 
removing a bridge does not change the optimum, while the optimum for a disconnected graph is the sum of the optima for its weakly connected components.
Let $b$ be the branchwidth of $G$.
By Theorem~\ref{thm:planarScDecomposition}, in polynomial time we can compute an sc-decomposition $(T,\eta)$ of $G$ of width $b$, which we fix from now on.

\medskip

We now describe the dynamic programming table.
For each $e \in E(T)$,  $X \subseteq \med(e)$, and $P \in \conn(\med(e)\setminus X)$,
define $\Val[e,X,P]$ as the minimum size of a vertex subset $S\subseteq V(G(e)) \setminus \med(e)$ such that $G_e - (X\cup S)$ 
is acyclic and induces the connectivity pattern $P$ on $\med(e)\setminus X$.
If no such set $S$ exists, define $\Val[e,X,P]$ as $\infty$.

The dynamic programming algorithm will compute all values of $\Val$ in a bottom-up manner.
It then suffices to return this value at the edge $e$ incident to the root: 
$\Val[e,\emptyset,\emptyset]$ is equal to the sought optimum for $G(e) = G$ (recall that $\med(e)=\emptyset$).

We remark that Theorem~\ref{thm:patterns} does not provide us any algorithm to enumerate all connectivity patterns in $\conn(\med(e)\setminus X)$, however this is not a problem.
Namely, in the algorithm we will store only those entries of $\Val$ that are different from $\infty$, together with connectivity patterns for which they are achieved, and for all the 
connectivity patterns that are not listed the value will be $\infty$. At any step of the computation we make sure that all the connectivity patterns from $\conn(\med(e)\setminus X)$
for which the relevant value $\Val$ is finite are explicitly constructed based on the connectivity patterns computed in the previous steps. Thus, it is not necessary to construct the
whole set $\conn(\med(e)\setminus X)$ in advance. We explain this formally at the end of the proof.

\medskip

We proceed to the description of the computation of the table entries, starting with decomposition edges incident to leaves of $T$.
For an edge $e$ of $T$ incident to a leaf corresponding (via $\eta$) to an edge $(u,v)$ of $G$, the value $\Val[e,\cdot,\cdot]$ can be easily computed as follows.
Namely, $\Val[e,X,P]$ equals $0$ when $X=\emptyset$ and $P$ is the pattern induced by $G(e)$, or when $X\neq \emptyset$ and $P$ is the only reflexive relation on $\{u,v\}\setminus X$, and
equals $\infty$ otherwise.

It remains to describe how to compute the entries for an edge $e$ after having computed the values at its children $e_1$ and $e_2$.
Let $P \in \conn(\med(e))$.
Then it is straightforward to see that $\Val[e,X,P]$ is equal to the minimum of $\Val[e_1,X_1,P_1] + \Val[e_2,X_2,P_2] + |Y|$ over all
$Y \subseteq (\med(e_1) \cup \med(e_2)) \setminus \med(e)$ and all $P_i \in \conn(\med(e_i)\setminus X_i)$ ($i=1,2$),
where $X_i := (X \cup Y) \cap \med(e_i)$, such that and $P$ is the join of $P_1$ and $P_2$. 

Indeed, let $A$ be the minimum defined above.
Take any subset $S$ of $V(G(e))\setminus \med(e)$ such that $G(e)-(X\cup S)$ is acyclic and induces the connectivity pattern $P$.
Let $Y=S\cap (\med(e_1) \cup \med(e_2)) \setminus \med(e))$, 
and let $P_i$ be the connectivity pattern on $\med(e_i)\setminus (X\cup Y)$ induced by $G(e_i)-(X\cup Y\cup S_i)$, for $i=1,2$, where $S_i:=(S\cap V(G(e_i)))\setminus Y$.
It follows that $S$ is the disjoint union of $S_1$, $S_2$, and $Y$, and each $G(e_i)-S_i$ is acyclic and 
induces some connectivity pattern $P_i$ on $\med(e_i)\setminus (X\cup Y)$ such that $P$ is the join of $P_1$ and $P_2$.
We infer that $A\geq \Val[e,X,P]$.
Conversely, if we have subsets $Y$, $S_1$, and $S_2$ as above, where each $G(e_i)-S_i$ is acyclic and induces a pattern $P_i$ on $\med(e_i)\setminus (X\cup Y)$ such that $P$ is the join of $P_1$ and $P_2$,
then $S:=S_1\cup S_2\cup Y$ is such that $G(e)-(X\cup S)$ is acyclic and induces the connectivity pattern $P$ on $\med(e)\setminus X$.
We infer that $A\leq \Val[e,X,P]$, so $A=\Val[e,X,P]$.

Recall that the algorithm stores, for each edge $e$ of the sc-decomposition, the list of all finite entries of the form $\Val[e,X,P]$, each together with the triple $(e,X,P)$.
Having computed these lists for children $e_1,e_2$ of an edge $e$, we compute the list for $e$ as follows.
For each $Y\subseteq (\med(e_1) \cup \med(e_2)) \setminus \med(e)$ and each 
pair of stored finite values $\Val[e_i,X_i,P_i]$, for $i=1,2$, $X_i$ as above, and such that $P_1$ and $P_2$ have a defined join $P$,
we memorize $\Val[e_1,X_1,P_1]+\Val[e_2,X_2,P_2]+|Y|$ as a candidate for $\Val[e,X,P]$.
Then the list for $e$ comprises all the entries $\Val[e,X,P]$ for which at least one candidate was found, and only the smallest candidate for each triple $(e,X,P)$ is stored on the list.
Since $|\med(e)|\leq b$ and $X\subseteq \med(e)$, there are at most $2^b$ ways to choose $X$.
Further, by Theorem~\ref{thm:patterns} there are at most $2^{\Oh(b)}$ ways to choose a connectivity pattern $P$ on $\med(e)\setminus X$, as there are only $2^{\Oh(b)}$ ways to choose a 
 relation $R$ generating $P$ from the family $\mathcal{R}$ provided by Theorem~\ref{thm:patterns}.
It follows that the lists of entries stored for each edge $e$ of the sc-decomposition will always have length bounded by $2^{\Oh(b)}$.
Hence computing the list for $e$ based on the previously computed lists for its children takes time $2^{\Oh(b)}$.
Since the size of the sc-decomposition is bounded linearly in the number of edges of $G$, it follows that the whole algorithm runs in time $2^{\Oh(b)} \cdot n^{\Oh(1)}$, as claimed.

\section{Dynamic programming on general digraphs}\label{sect:pfthm1}

\newcommand{\treedecomp}{\ensuremath{T}}

In this section we give a full exposition of the proof of Theorem~\ref{thm:simple-dp}. 
To this end, we first recall a number of standard definitions and observations.

A \emph{tree decomposition} of a graph~$G=(V,E)$ is a tree~$\treedecomp$ in which each 
node~$x$ has an assigned set of vertices~$B_x \subseteq V$ (called a \emph{bag}) such that $\bigcup_{x \in \treedecomp} B_x = V$ and the following assertions hold:
\begin{itemize}[nosep]
\item for any $uv \in E$, there exists a node~$x \in \treedecomp$ such that 
$u,v \in B_x$.
\item if $v \in B_x\cap B_y$ for some nodes $x,y$ of $T$, then $v \in B_z$ for each node $z$ on 
the (unique) path from $x$ to $y$ in $\treedecomp$.
\end{itemize}
The \emph{width} of a tree decomposition~$\treedecomp$ is the size of 
the largest bag of $\treedecomp$ minus one, and the {\em{treewidth}} of a graph $G$ is the 
minimum {\em{treewidth}} over all possible tree decompositions of~$G$.

If the treewidth of a graph $G$ is $t$, then a tree decomposition of $G$ of width at most $5t+4=\Oh(w)$ can be computed in time $2^{\Oh(t)}\cdot n$~\cite{BodlaenderDDFLP16}.
Hence, from now on we may assume that we are given a digraph $G$ together with its tree decomposition of width $t$, and the goal is to give a dynamic programming algorithm with running time $2^{\Oh(t)}\cdot n^{\Oh(1)}$.

A {\em{rooted}} tree decomposition is a tree decomposition that is rooted at one of its nodes, called the {\em{root}}; this naturally imposes the child/parent relation on the nodes of the decompositions.
{\em{Nice}} tree decompositions are tree decompositions normalized for the purpose of streamlining the presentation of dynamic programming algorithms.

\begin{definition}[Nice Tree Decomposition] \label{def:nicetreedecomp}
A \emph{nice tree decomposition} is a rooted tree decomposition in which every node is one of the following types:
\begin{itemize}[nosep]
\item \textbf{Leaf node}: a leaf node $x$ of $\treedecomp$ with $B_x = \emptyset$.
\item \textbf{Introduce node}: an internal node~$x$ of $\treedecomp$ 
with one child vertex~$y$ for which $B_x = B_y \cup \{v\}$ 
for some $v \notin B_y$. 
This node is said to \emph{introduce} $v$.
\item \textbf{Forget node}: an internal node~$x$ of $\treedecomp$ with one child 
node~$y$ for which $B_x = B_y \setminus \{v\}$ for some $v \in B_y$. 
This node is said to \emph{forget} $v$.
\item \textbf{Join node}: an internal node $x$ with two children 
$y$ and $z$ satisfying $B_x = B_y = B_z$.
\end{itemize}
Moreover, if $r$ is the root node of the tree decomposition, then $B_r=\emptyset$.
\end{definition}
It is well known that given a tree decomposition of $G$, 
a nice tree decomposition of $G$ of the same width can be constructed in polynomial time. 
We refer to ~\cite{ParametrizedAlgorithmsBook} for further details.
Hence, from now on we may assume that the given tree decomposition $\treedecomp$ is nice.

By fixing the root of $\treedecomp$, we associate with each node $x$ 
in a tree decomposition $\treedecomp$ a vertex set $V_x \subseteq V$, where a vertex 
$v$ belongs to $V_x$ if and only if there is a bag $y$
which is a descendant of $x$ in $\treedecomp$ with $v \in B_y$
(recall that we treat $x$ as its own descendant as well).
We also associate with each bag $x$ of $\treedecomp$ the subgraph $G_x$ of $G$ induced by $V_x$, i.e., $G_x=G[V_x]$.

We use dynamic programming to solve \ProblemName{Directed Feedback Vertex Set}.
To this end, for every node $x \in V(T)$, every subset $X \subseteq B_x$, and every ordering $\sigma$ of $B_x\setminus X$, 
we define the dynamic programming table as follows: $T_x[X,\sigma]$ is the minimum size of a subset $Y$ of $V(G_x)\setminus B_x$ such that $G_x-(X\cup Y)$ is acyclic and admits a topological ordering 
whose restriction to $B_x\setminus X$ is exactly $\sigma$. That is, we define $T_x[X,\sigma]$ as
\[
	\displaystyle \min \{\ |Y|\,\colon\,  Y \subseteq V_x\setminus B_x,\ \sigma' \text{ is a topological ordering of } G_x-(X\cup Y),\text{ and }\sigma'|_{B_x\setminus X}=\sigma\ \},
\]
where we write $\pi|_X$ for the restriction of $\pi$ to the elements of $X$. If there is no such set $Y$, we put $T_x[X,\sigma]=+\infty$.

It is clear that if $r$ is the root node of $\treedecomp$, then the minimum size of a directed feedback vertex set in $G$ is equal to the only value associated with $r$, namely $T_r[\emptyset,\epsilon]$,
where $\epsilon$ denotes the empty ordering
(recall here that the bag of the root node is empty). 
Hence, it remains to compute all the values $T_x[\cdot,\cdot]$ for nodes $x$ of $\treedecomp$. We do it in a bottom-up manner using the following recursive
equations.
\smallskip
\begin{itemize}[nosep]
	\item \textbf{Leaf node:} If $x$ is a leaf node, we see that $G_x$ is an empty graph so $T_x[X,\sigma]$ is only defined for $X=\emptyset$ and $\sigma=\epsilon$ being the empty ordering, 
	                          and $T_x[\emptyset,\epsilon]=0$. 
	\item \textbf{Introduce node:} Suppose $x$ introduces a vertex $v$. If $v \in X$, we have
	\[
		T_x[X,\sigma] = T_y[X\setminus \{v\},\sigma],
	\]
	as $v\in X$ is equivalent to removing $v$ from $G_x$ in the definition of $T_x[\cdot,\cdot]$ anyway.
	On the other hand, if $v\notin X$, we have
	\[
		T_x[X,\sigma] =
 		\begin{cases}
			T_y[X,\sigma|_{B_y\setminus X}], & \text{if } \sigma \text{ orders } N^+_X(v) \cap B_x \text{ after } v \text{ and } N^-_X(v) \cap B_x \text{ before } v,\\
			+\infty & \text{otherwise}.
		\end{cases}
	\]
	Here we use $N^+_X(v)$ and $N^-_X(v)$ to denote the set of out-neighbors and respectively in-neighbors of $v$ in $X$.
	To see this last case, note that if $v \notin X$, 
	it needs to be ordered by $\sigma$, but the topological order of all its in- and out-neighbors in $G_x-X$ (which all lie in $B_y\setminus X$) 
	is already fixed to be as in $\sigma$, so we only need to check that $v$ fits into this ordering.
	\item \textbf{Forget node:} Suppose $x$ forgets a vertex $v$. Then we have
	\[
		T_x[X,\sigma] = \min \{ T_y[X \cup \{v\},\sigma]+1, \min\{ T_y[X,\sigma']\colon \sigma'\text{ extends }\sigma\} \},
	\]
	as in $T_x[X,\sigma]$ we minimize over all sets $Y$ containing $v$ and all sets $Y$ avoiding~$v$. In the second case, vertex $v$ is not removed and hence has to be ordered; hence we consider
	all orderings $\sigma'$ of $B_y\setminus X=(B_x\setminus X)\cup \{v\}$ that {\em{extend}} $\sigma$: they differ from $\sigma$ by placing $v$ somewhere in the order.
	
	\item \textbf{Join node:} For a node $x$ with children $y$ and $z$ we have
	\[
		\displaystyle T_x[X,\sigma] = T_y[X,\sigma] + T_z[X,\sigma].
	\]
	To see this, first note that any topological ordering $\sigma$ of $G_x-(X\cup Y)$ gives rise to topological orderings $\sigma_y = \sigma|_{V(G_y) \setminus (X \cup Y)}$ of $G_y-(X\cup Y)$
	and $\sigma_z = \sigma|_{V(G_z) \setminus (X \cup Y)}$ of $G_z-(X\cup Y)$. As $Y$ is the disjoint union of $Y\cap V(G_y)$ and $Y\cap V(G_z)$, this proves that $T_x[X,\sigma]\leq T_y[X,\sigma]+T_z[X,\sigma]$.
	For the reverse inequality, suppose $\sigma_y$ and $\sigma_z$ are topological orderings of $G_y - (X \cup Y)$ and $G_z - (X \cup Y)$ respectively, and that they agree on $B_x$/ 
	Then they can be combined into a topological ordering of $G_x - (X \cup Y)$ as the relative order of vertices in $V(G_y) \setminus B_x$ and $V(G_z)\setminus B_z$ is immaterial.
\end{itemize}
\smallskip
Since for each bag $B_x$ of the tree decomposition $\treedecomp$ is of size $|B_x|\leq t+1$, 
we compute at most $2^t \cdot t! = 2^{\Oh (t \log t)}$ values of the form $T_x[\cdot,\cdot]$ for each node $x$ of the decomposition.
Moreover, the computation of each value takes time polynomial in $t$, so the overall running time of the algorithm is $2^{\Oh(t\log t)}\cdot n^{\Oh(1)}$, as claimed.
This concludes the description of the algorithm for \ProblemName{Directed Feedback Vertex Set}.

For \ProblemName{Directed Feedback Arc Set} the dynamic programming is analogous, only simpler in that we do not need to consider the vertex subset $X$.
We just define, for each $x\in V(T)$ and each ordering $\sigma$ of $B_x$, the value $T_x[\sigma]$ as the minimum size of a subset $Y \subseteq E(G_x)$ 
such that $G_x - Y$ is acyclic and admits a topological ordering whose restriction to $B_x$ is exactly $\sigma$.
It is straightforward to adapt the recursive equations given above to this definition of the dynamic programming table.
This concludes the proof of Theorem~\ref{thm:simple-dp}.

\section{Lower bound}\label{sect:lowerBound}
In this section we give slightly super-exponential lower bounds under ETH for {\sc{Directed Feedback Arc Set}} and {\sc{Directed Feedback Vertex Set}}  parameterized by treewidth in general digraphs, that is, we prove 
Theorem~\ref{thm:main-lb}.
The hardness reduction happens to work for both problems, producing exactly the same instances.

We will reduce from the following problem, whose hardness was shown by Lokshtanov et al.~\cite{LokshtanovMS11} (see also~\cite[Theorem 14.16]{ParametrizedAlgorithmsBook}).

\defproblempar{\ProblemName{$k\times k$ Hitting Set with thin sets}}
{Family $\Ff$ of subsets of $[k]\times[k]$, each containing at most one element from each~row}
{$k$}
{Is there a set $X$ containing one vertex from each row of $[k]\times[k]$ such that $X \cap F \neq \emptyset$ for each $F \in \Ff$?}

\begin{theorem}[\cite{LokshtanovMS11}]\label{thm:hittingSet}
	Unless ETH fails, \ProblemName{$k\times k$ Hitting Set with thin sets} cannot be solved in time $2^{o(k \log k)}\cdot n^{\Oh(1)}$, where $n$ is the number of input sets.
\end{theorem}

\pagebreak[3]
Let us first define an intermediate problem. An \emph{$n$-permutation $d$-constraint} is a tuple $(i_1,\dots,i_d) \in [n]^d$ of $d$ different indices.
A permutation $\sigma\colon [n] \to [n]$ \emph{satisfies} such a constraint if $\sigma(i_1) < \sigma(i_2) < \cdots < \sigma(i_d)$.
A \emph{$k$-CNF $n$-permutation $d$-formula} is a conjunction of clauses, each of which is a disjunction of at most $k$ $n$-permutation $d$-constraints.
The {\em{length}} of a clause is the number of disjuncts (constraints) in it.
Satisfaction of such a formula by a permutation $\sigma\colon [n] \to [n]$ is defined naturally.


We first show hardness for the satisfiability of 3-formulas, with the parameter $k$ denoting both the length of clauses and the number of indices on which the permutation is defined.

\begin{lemma}\label{lem:3PermutationSat}
	Unless ETH fails, the satisfiability of a given $k$-CNF $k$-permutation 3-formula cannot be decided in time $2^{o(k \log k)}\cdot n^{\Oh(1)}$, where $n$ is the size of the formula. 
\end{lemma}
\begin{proof}
        Without loss of generality suppose $k\geq 3$.
	Let $\Ff$ be the input instance of \ProblemName{$k\times k$ Hitting Set with thin sets}.
	We will construct in polynomial time a $k$-CNF $(2k+1)$-permutation $3$-formula whose satisfiability is equivalent to the input instance $\Ff$.
	This will prove the claim by Theorem~\ref{thm:hittingSet}.

	To an initially empty formula $\phi$ we add the following clauses, each with a single 3-constraint, to ensure that $\{k+1,\dots,2k+1\}$ are ordered increasingly by the permutation: $$(k+1,k+2,k+3),(k+2,k+3,k+4),\ldots, (2k-1,2k,2k+1).$$
	Then, for each $i \in [k]$ we add a clause with a single 3-constraint $(k+1,i,2k+1)$.
	Finally, for each set $F\in \Ff$, we add the following clause $C_F$ to~$\phi$:
	the clause $C_F$ is the disjunction of constraints $(k+j,i,k+j+1)$ over elements $(i,j)$ of $F$.
	Since $F$ contains at most one element of each row, the clause $C_F$ is a disjunction of at most $k$ constraints.
	This concludes the construction.

	Suppose the input instance $\Ff$ has a solution $X \subseteq [k]\times [k]$.
	We define a permutation $\sigma\colon [2k+1] \to [2k+1]$ satisfying $\phi$ as follows.
	First, indices $\{k+1,\dots,2k+1\}$ are ordered increasingly.
	For each of the remaining indices $i \in [k]$, let $(i,j) \in X$ be the unique element of $X$ in row $i$.
	Then order $i$ to be between $k+j$ and $k+j+1$, arbitrarily choosing the ordering among other indices that were also placed between $k+j$ and $k+j+1$.
	This ordering defines the permutation $\sigma$, which clearly satisfies the formula.
	To see this, note that for every $F\in \Ff$, if $(i,j) \in X \cap F$ then the constraint $(k+j,i,k+j+1)$ is in $C_F$ due to $(i,j)\in F$, and it is satisfied by $\sigma$ due to $(i,j)\in X$.

	Conversely, suppose $\phi$ is satisfied by some permutation $\sigma\colon [2k+1] \to [2k+1]$.
	Then $\sigma$ orders $\{k+1,\dots,2k+1\}$ increasingly and each of the remaining indices $i \in [k]$ is ordered between $k+1$ and $2k+1$.
	Hence, for each $i\in [k]$ there is a unique index $j_i\in [k]$ such that $i$ is ordered between $k+j_i$ and $k+j_i+1$.
	Let $X = \{ (i,j_i) \colon i \in [k]\}$. Since $\sigma$ satisfies $\phi$, for each $F\in\Ff$ some constraint in the clause $C_F$ is satisfied by $\sigma$, 
	and this clause must be of the form $(k+j,i,k+j+1)$ for some $(i,j) \in F$.
	As $j_i$ is the unique index such that $i$ is ordered between $k+j_i$ and $k+j_i+1$, it must be that $(i,j_i) \in F$ for some $i\in [k]$. 
	Therefore $X\cap F \neq \emptyset$, so $X$ is a solution to the input instance.
\end{proof}

Next, we show hardness for larger, but structured 2-formulas.
For a $3$-CNF $n$-permutation $2$-formula $\phi$, the \emph{incidence graph} $I(\phi)$ of $\phi$ is the bipartite graph defined as follows: the vertex set is formed by 
indices from $[n]$ on one side and clauses of $\phi$ on the other side, 
and there is an edge between every clause and each index that occurs in some constraint of the clause. Thus, each clause has degree at most $6$ in $I(\phi)$.

\begin{lemma}\label{lem:2PermutationSat}
	Unless ETH fails, the satisfiability of a given 3-CNF $n$-permutation 2-formula with incidence graph of treewidth $t$ cannot be decided in time $2^{o(t\,\log t)}\cdot n^{\Oh(1)}$.
	This holds even for formulas in which every clause has length exactly 3 or 1, and has no repeating indices.
\end{lemma}
\begin{proof}
	Let $\phi$ be a $k$-CNF $k$-permutation 3-formula.
	We will construct  in polynomial time a 3-CNF $n$-permutation 2-formula $\psi$ for some $n=\Oh(k^2)$ such that
	$\psi$ is satisfiable iff $\phi$ is and the incidence graph of $\psi$ has treewidth $\Oh(k)$. 
	The claim then follows by Lemma~\ref{lem:3PermutationSat}. 

	The idea is that every 3-constraint $(a,b,c)$ can be thought of as a conjunction $(a,b) \wedge (b,c)$ of two 2-constraints (expressing $\sigma(a)<\sigma(b)\ \wedge\ \sigma(b)< \sigma(c)$).
	Intuitively, we can then transform the obtained `non-CNF formula' into a 3-CNF in a standard way: a clause $(x \wedge x') \vee (y \wedge y') \vee (z \wedge z') \vee \dots$ would be replaced by
	\begin{align*}
		(p_1)\ \wedge\ (\neg p_1 \vee x \vee p_2)\ &\wedge\ (\neg p_2 \vee y \vee p_3)\ \wedge\ (\neg p_3 \vee z \vee p_4)\ \wedge\ \dots\\
		      \wedge\ (\neg p_1 \vee x'\vee      p_2)\ &\wedge\ (\neg p_2 \vee y'\vee p_3)\ \wedge\ (\neg p_3 \vee z'\vee p_4)\ \wedge\ \dots\ \wedge\ (\neg p_n)
	\end{align*}
	where $p_1,p_2,p_3,\dots,p_n$ are fresh auxiliary variables that do not appear anywhere else.

	Formally, we will ask for $n$-permutations with $n := k+(2k+2)k$; the additional indices are in order to make room for `auxiliary variables'.
	We construct $\psi$ as an initially empty conjunction.
	Each clause $C$ of $\phi$ is a disjunction $C_1 \vee \dots \vee C_{k'}$  ($k' \leq k$) of some 3-constraints $C_i = (a_i,b_i,c_i) \in [k]^3$.
	Let $j_1,j_2,\dots,j_{2k'+2} \in [n]\setminus[k]$ be some indices that were not yet used in any constructed clause.
	For each $i\in [k']$, we add the following clauses $D_i$ and $D_i'$ to $\psi$:
	\begin{align*}
	 D_i & = (j_{2i},j_{2i-1}) \vee (a_i,b_i) \vee (j_{2i+1},j_{2i+2})\\
	 D_i' & = (j_{2i},j_{2i-1}) \vee (b_i,c_i) \vee (j_{2i+1},j_{2i+2})
	\end{align*}
	Then, we add two clauses with a single constraint each: $Z=(j_1,j_2)$ and $Z'=(j_{2k'+2},j_{2k'+1})$.
	Repeating this for each clause $C$ of $\phi$ concludes the construction.
	Let $W(C)$ be the set consisting of clauses and indices used for $C$: clauses $Z,Z'$, clauses $D_i,D_i'$ for each $i\in [k']$, and indices $j_1,j_2,\ldots,j_{2k'+2}$ as above.
	Then $[k]$ together with sets $W(C)$ for clauses $C$ of $\phi$ form a partition of the vertex set of the incidence graph $I(\psi)$ of the constructed formula.

	We first bound the treewidth of $I(\psi)$.
	Observe that in $I(\psi)$, if we remove all the $k$ vertices corresponding to $[k]\subseteq [n]$, the only remaining edges have both endpoints within the same $W(C)$ for some clause $C$ of $\phi$.
	Since each $W(C)$ has size at most $3k+4$, each connected component of the remaining graph has size at most $3k+4 = \Oh(k)$.
	Therefore $I(\psi)$ admits a very simple tree decomposition of width $\Oh(k)$: the shape of the decomposition is a star with $[k]$ as the bag at the center,
	and each petal of the star corresponds to one clause $C$ of $\phi$---its bag is $[k]\cup W(C)$.
	
	It remains to prove that $\phi$ is satisfiable if and only if $\psi$ is.
	First, given a permutation $\sigma\colon [k]\to[k]$ satisfying $\phi$, we can construct a permutation $\pi\colon [n] \to [n]$ satisfying $\psi$ as follows.
	The indices in $[k]$ are ordered in the same way as in $\sigma$.
	To order the indices $j_1,j_2,\dots,j_{2k'+2} \in [n]\setminus[k]$ used for a clause $C$ of $\phi$,
	let $C_{i_0} = (a_{i_0},b_{i_0},c_{i_0})$ be the constraint satisfied by $\phi$ in this clause. Then we order 
	$j_{2i-1}$ before $j_{2i}$ for $1 \leq i \leq i_0$, and $j_{2i}$ before $j_{2i-1}$ for $i > i_0$.
	All other ordering choices are arbitrary.
	Observe that clauses $D_{i_0}, D_{i_0}'$ are satisfied because their middle constraint $ (a_{i_0},b_{i_0},c_{i_0})$ is satisfied, 
	clauses $D_i, D_i'$ for $i<i_0$ are satisfied by their right constraint $(j_{2i+1},j_{2i+2})$, 
	while clauses $D_i, D_i'$ for $i>i_0$ are satisfied by their left constraint $(j_{2i},j_{2i-1})$.
	This proves that $\pi$ indeed satisfies $\psi$.
	
	Conversely, let $\pi\colon[n] \to [n]$ be a permutation satisfying the constructed formula $\psi$.
	Then we claim the permutation $\sigma\colon[k] \to [k]$ that orders $[k]$ in the same way as $\pi$ satisfies the original formula $\phi$. 
	Indeed, take any clause $C=C_1 \vee \dots \vee C_{k'}$ of $\phi$ and consider the ordering in $\pi$ of each pair of the corresponding indices $(j_1, j_2), (j_3,j_4), \dots, (j_{2k'+1},j_{2k'+2})$.
	The first of these pairs must be ordered increasingly, while the last one must be ordered decreasingly; this is due to the clauses $Z,Z'$ introduced for $C$. 
	Hence we may define $i_0 \in [k']$ to be the last index such that $(j_{2i_0-1},j_{2i_0})$ is ordered increasingly in $\pi$; that is, $\pi(j_{2i_0-1})<\pi(j_{2i_0})$.
	Then  $(j_{2i_0+1},j_{2i_0+2})$ is ordered decreasingly, hence the clauses $D_{i_0}$ and $D_{i_0}'$ are not satisfied by their left and right constraints. 
	Therefore, they are both satisfied by their middle constraints, that is, $\pi(a_{i_0}) < \pi (b_{i_0})$ and $\pi(b_{i_0}) < \pi(c_{i_0})$. 
	This means $\pi$ satisfies the 3-constraint $(a_{i_0},b_{i_0},c_{i_0})$. 
	Since $\sigma$ orders $a_{i_0}, b_{i_0}, c_{i_0} \in [k]$ in the same way, it satisfies this 3-constraint as well, hence it satisfies the clause $C$ that contains it.
\end{proof}

We proceed to reducing to satisfiability of permutation formulas as described in Lemma~\ref{lem:2PermutationSat} to {\sc{Directed Feedback Vertex (Arc) Set}}.
Permutations of $[n]$ will be encoded as orderings of a subset of $n$ `terminal' vertices in the graph constructed by the reduction.
The graph will contain gadgets that ensure that the permutation satisfies the original 3-CNF $n$-permutation 2-formula if and only if the ordering can be extended to a topological ordering of the whole graph, 
after deleting a prescribed number of vertices (edges). 
The key element is the or-gadget depicted in Figure~\ref{fig:gadget}, which encodes a clause that is a disjunction of three 2-constraints.
Note that this or-gadget has 6 terminal vertices, named $x_i,x_i'$ for $i \in [3]$.
The final graph is obtained essentially by taking disjoint copies of the or-gadget and identifying their terminal vertices with indices from $[n]$.

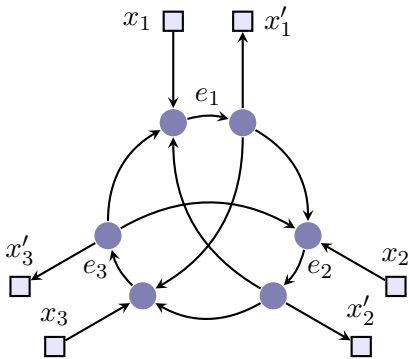
\begin{figure}[b!]
	\centering
	\begin{tikzpicture}[->,>=stealth,thick,scale=0.9]
	\tikzstyle{v} = [white!25!gray!80!blue,circle,fill,minimum size=0.5]
	\tikzstyle{t} = [black,rectangle,draw,fill=blue!10!white,minimum size=1]
	\node[v] (v1a) at (110:1.5) {};
	\node[v] (v1b) at (70:1.5) {};
	\node[v] (v2a) at (-10:1.5) {};
	\node[v] (v2b) at (-50:1.5) {};
	\node[v] (v3a) at (-130:1.5) {};
	\node[v] (v3b) at (-170:1.5) {};
	\node[t,label=left:$x_1$]   (x1a) at ($(v1a)+(90:1.5)$) {};
	\node[t,label=right:$x_1'$] (x1b) at ($(v1b)+(90:1.5)$) {};
	\node[t,label=above:$x_2$] (x2a) at ($(v2a)+(-30:1.5)$) {};
	\node[t,label=above:$x_2'$] (x2b) at ($(v2b)+(-30:1.5)$) {};
	\node[t,label=above:$x_3$] (x3a) at ($(v3a)+(-150:1.5)$) {};
	\node[t,label=above:$x_3'$] (x3b) at ($(v3b)+(-150:1.5)$) {};
	\draw (x1a)--(v1a); \draw (v1b)--(x1b);
	\draw (x2a)--(v2a);	\draw (v2b)--(x2b);
	\draw (x3a)--(v3a);	\draw (v3b)--(x3b);
	\draw (v1a) edge[bend left=15] node[above]{$e_1$} (v1b)  (v1b) edge[bend left] (v2a)
    (v2a) edge[bend left=15] node[right]{$e_2$} (v2b)  (v2b) edge[bend left] (v3a)
    (v3a) edge[bend left=15] node[left]{$e_3$} (v3b)  (v3b) edge[bend left] (v1a);
	\draw (v1b) edge[bend left] (v3a)
    (v2b) edge[bend left] (v1a)
    (v3b) edge[bend left] (v2a);
\end{tikzpicture}
	\caption{The or-gadget, with 6 terminal vertices $x_i,x_i'$ ($i \in [3]$) marked as squares.}
	\label{fig:gadget}
\end{figure}

\begin{lemma}\label{lem:inspectOrGadget}
	For an ordering $\prec$ of the terminal vertices of the or-gadget, $\prec$ can be extended to a topological ordering of the or-gadget with some 2 vertices (edges) deleted 
	if and only if $x_1 \prec x_1'$ or $x_2 \prec x_2'$ or $x_3 \prec x_3'$. 
	Furthermore, every subgraph of the or-gadget obtained by deleting at most one non-terminal vertex or an edge from it, contains a directed cycle.
\end{lemma}
\begin{proof}
	Given an ordering $\prec$ of the terminal vertices such that $x_1 \prec x_1'$, one can remove $e_2$ and $e_3$, or any two vertices incident to them, to create an acyclic subgraph of the or-gadget that
	admits a topological ordering extending $\prec$. The cases of orderings $\prec$ with $x_2\prec x_2'$ and with $x_3\prec x_3'$ are symmetric.
	Conversely, any removal of two vertices or edges from the or-gadget leaves some directed path $x_i \to x_i'$ ($i \in [3]$) unharmed, implying $x_i \prec x_i'$ in any topological ordering of the obtained subgraph. Finally, it is easy to check that at least two non-terminal vertices or edges of the or-gadget have to be removed to make it acyclic.
\end{proof}

We are ready to conclude our reduction and prove Theorem~\ref{thm:main-lb}.

\begin{proof}[Proof of Theorem~\ref{thm:main-lb}]
	We give a reduction from the satisfiability problem for 3-CNF $n$-permutation 2-formulas to {\sc{DFVS}} and {\sc{DFAS}}.
	More precisely, on the input of the reduction we are given a 3-CNF $n$-permutation 2-formula $\psi$ with an incidence graph of treewidth $t$,
	where we assume that every clause of $\psi$ has length exactly 3 or 1, and has no repeating indices.
	We will construct in polynomial time an equivalent instance of (the decision variant of) {\sc{DFVS}} ({\sc{DFAS}}) of treewidth $\Oh(t)$.
	This will prove the claim by Lemma~\ref{lem:2PermutationSat}.

	We construct a digraph $G$ starting from $[n]$ as the vertex set and no edges.
	For each clause of length 1 in $\psi$, let $(a,a') \in [n]^2$ be the unique constraint in it.
	Then we add an edge from $a$ to $a'$ to $G$.
	For each clause of length 3 in $\psi$, let $(a_1,a_1'), (a_2,a_2'), (a_3,a_3') \in [n]^2$ be its constraints.
	Then we add a new copy of the or-gadget to $G$, and for each $i\in [3]$ we identify $x_i$ and $x_i'$ with $a_i$ and $a_i'$, respectively.
	Finally, we set $k$ to be twice the number of clauses of length 3 in $\psi$.
	The obtained instance $(G,k)$ will be treated both as a {\sc{DFVS}} instance and as a {\sc{DFAS}} instance.
	We proceed with the reasoning for {\sc{DFVS}} and {\sc{DFAS}} in parallel, giving always the counterpart for {\sc{DFAS}} in parentheses.
	
	Let us first bound the treewidth of $G$.
	Observe that $G$ can be obtained from $I(\psi)$ by the following replacements. First, each vertex $w$ representing 
	a clause of length 3 (with 6 neighbors in $[n]$) is replaced by a copy of the or-gadget, consisting of $6$ nonterminal vertices connected via $6$ edges to the former neighbors of $w$, which are identified
	with the terminal vertices of the gadget.
	Second, each vertex $w$ representing a clause of length 1 (with 2 neighbors in $[n]$) is replaced by an edge connecting the former neighbors of $w$.
	It is easy to see that the first operation can only increase the treewidth by multiplicative factor at most 6, while the second operation can only decrease the treewidth. 
	Therefore $G$ has treewidth $\Oh(t)$.

	It remains to show that the instances are equivalent, which boils down to applying Lemma~\ref{lem:inspectOrGadget}.
	Let $\pi\colon [n] \to [n]$ be a permutation satisfying $\psi$.
	Consider the ordering $\prec$ of $[n] \subseteq V(G)$ that is induced by $\pi$; i.e., $i\prec j$ iff $\pi(i)<\pi(j)$.
	We will extend $\prec$ to a topological ordering of $G-X$, for some set $X$ of $k$ vertices (edges).
	The edges introduced for clauses of length 1 are already oriented in the right way by the construction.
	Hence, it suffices to extend $\prec$ to a topological ordering of each copy of the or-gadget independently, since $G$ has no other edges.
	Consider the copy introduced for a clause $C = (a_1,a_1') \vee (a_2,a_2') \vee (a_3,a_3') $ of $\psi$.
	Recall that the terminal vertices of this copy have been identified with $a_i, a_i' \in [n]$ accordingly.
	Suppose without loss of generality that $(a_1,a_1')$ is satisfied by $\pi$, that is, $a_1\prec a_2$.
	Then by Lemma~\ref{lem:inspectOrGadget}, it suffices to add 2 vertices (edges) to the deletion set $X$ to allow extending the ordering on $a_i, a_i'$ for $i\in [n]$ to a topological ordering of the whole or-gadget.
	In total, we add 2 vertices (edges) to $X$ for every clause of length 3 in $\psi$, which yields $k$ vertices (edges) in total.
	Since $G-X$ has a topological ordering (with the order between vertices from different or-gadgets chosen arbitrarily), it has no directed cycles, so $X$ is a solution.

	Conversely, suppose $X$ is a set of size at most $k$ such that $G-X$ has no directed cycles.
	Let $\pi$ be a topological ordering of $G$.
	We claim $\pi|_{[n]}$---the ordering $\pi$ restricted to $[n]$---defines a permutation that satisfies $\psi$.
	Since each copy of the or-gadget requires deleting at least two non-terminal vertices (edges), by~Lemma~\ref{lem:inspectOrGadget}, 
	and $|X|\leq k$, $X$ has to contain exactly two non-terminal vertices (edges) from each or-gadget.
	In particular, $X$ contains no vertex from $[n]$ or edge with both endpoints in $[n]$, so none of the edges introduced for clauses of length 1 (between vertices of $[n]$) are deleted by $X$. 
	Hence $\pi|_{[n]}$ satisfies each such clause by the construction.
	Finally, for each clause $C$ of length 3 in $\psi$, consider the copy of the or-gadget introduced in $G$ for $C$.
	Since $X$ contains exactly two non-terminal vertices (edges) from this gadget and 
	$\pi$ yields a topological ordering of this gadget with $X$ removed, by Lemma~\ref{lem:inspectOrGadget} it follows that $\pi|_{[n]}$ must satisfy $C$.
\end{proof}

\section{Conclusions}\label{sect:conclusions}
We investigated the complexity of {\sc{Directed Feedback Vertex Set}} parameterized by the treewidth $t$ of the input digraph.
We proved that in general digraphs there is an algorithm with running time $2^{\Oh(t\log t)}\cdot n^{\Oh(1)}$, which is optimal under the Exponential Time Hypothesis.
On the other hand, in planar digraphs the running time can be improved to $2^{\Oh(t)}\cdot n^{\Oh(1)}$.

Our results do not provide any direct insight into the complexity of the classic parameterization: by the target solution size $k$.
We hope, however, that the combinatorial tools we used in the proof of Theorem~\ref{thm:main-algo} may be useful for improving the running time for {\sc{DFVS}} on planar digraphs, say to running time
$2^{\Oh(k)}\cdot n^{\Oh(1)}$, or for obtaining a somewhat incomparable running time $n^{\Oh(\sqrt{k})}$. Observe that there is a large gap between known results in this setting: 
while the classic reduction from {\sc{Vertex Cover}} on planar graphs gives a lower bound excluding
running time $2^{o(\sqrt{k})}\cdot n^{\Oh(1)}$ under ETH, no faster algorithm than $2^{\Oh(k\log k)}\cdot (n+m)$ from general digraphs~\cite{LokshtanovRS16} is known.

\bibliography{planar-dfvs}

\begin{thebibliography}{10}

\bibitem{BianGZ16}
Z.~Bian, Q.~Gu, and M.~Zhu.
\newblock Practical algorithms for branch-decompositions of planar graphs.
\newblock {\em Discrete Applied Mathematics}, 199:156--171, 2016.

\bibitem{BodlaenderDDFLP16}
H.~L. Bodlaender, P.~G. Drange, M.~S. Dregi, F.~V. Fomin, D.~Lokshtanov, and
  M.~Pilipczuk.
\newblock A $c^k \cdot n$ 5-approximation algorithm for treewidth.
\newblock {\em {SIAM} J. Comput.}, 45(2):317--378, 2016.

\bibitem{ChenLLOR08}
J.~Chen, Y.~Liu, S.~Lu, B.~O'Sullivan, and I.~Razgon.
\newblock A fixed-parameter algorithm for the {D}irected {F}eedback {V}ertex
  {S}et problem.
\newblock {\em J. {ACM}}, 55(5):21:1--21:19, 2008.

\bibitem{ChitnisCHM15}
R.~H. Chitnis, M.~Cygan, M.~T. Hajiaghayi, and D.~Marx.
\newblock Directed {S}ubset {F}eedback {V}ertex {S}et is fixed-parameter
  tractable.
\newblock {\em {ACM} Trans. Algorithms}, 11(4):28:1--28:28, 2015.

\bibitem{ChitnisHM13}
R.~H. Chitnis, M.~Hajiaghayi, and D.~Marx.
\newblock Fixed-parameter tractability of {D}irected {M}ultiway {C}ut
  parameterized by the size of the cutset.
\newblock {\em {SIAM} J. Comput.}, 42(4):1674--1696, 2013.

\bibitem{ParametrizedAlgorithmsBook}
M.~Cygan, F.~V. Fomin, L.~Kowalik, D.~Lokshtanov, D.~Marx, M.~Pilipczuk,
  M.~Pilipczuk, and S.~Saurabh.
\newblock {\em Parameterized Algorithms}.
\newblock Springer, 2015.

\bibitem{Diestel2010GraphTheoryBook}
R.~Diestel.
\newblock {\em Graph Theory}.
\newblock Springer, 2010.

\bibitem{DornPBF10}
F.~Dorn, E.~Penninkx, H.~L. Bodlaender, and F.~V. Fomin.
\newblock Efficient exact algorithms on planar graphs: Exploiting sphere cut
  decompositions.
\newblock {\em Algorithmica}, 58(3):790--810, 2010.

\bibitem{FominT06}
F.~V. Fomin and D.~M. Thilikos.
\newblock New upper bounds on the decomposability of planar graphs.
\newblock {\em Journal of Graph Theory}, 51(1):53--81, 2006.

\bibitem{GuT08}
Q.~Gu and H.~Tamaki.
\newblock Optimal branch-decomposition of planar graphs in ${O}(n^3)$ time.
\newblock {\em {ACM} Trans. Algorithms}, 4(3):30:1--30:13, 2008.

\bibitem{Gyarfas85}
A.~Gy{\'{a}}rf{\'{a}}s.
\newblock On the chromatic number of multiple interval graphs and overlap
  graphs.
\newblock {\em Discrete Mathematics}, 55(2):161--166, 1985.

\bibitem{Gyarfas86}
A.~Gy{\'{a}}rf{\'{a}}s.
\newblock Corrigendum: On the chromatic number of multiple interval graphs and
  overlap graphs.
\newblock {\em Discrete Mathematics}, 62(3):333, 1986.

\bibitem{gyarfas1987problems}
A.~Gy{\'a}rf{\'a}s.
\newblock Problems from the world surrounding perfect graphs.
\newblock {\em Applicationes Mathematicae}, 19(3--4):413--441, 1987.

\bibitem{ImpagliazzoPZ01}
R.~Impagliazzo, R.~Paturi, and F.~Zane.
\newblock Which problems have strongly exponential complexity?
\newblock {\em J. Comput. Syst. Sci.}, 63(4):512--530, 2001.

\bibitem{KimG13}
E.~J. Kim and D.~Gon{\c{c}}alves.
\newblock On exact algorithms for the permutation {CSP}.
\newblock {\em Theor. Comput. Sci.}, 511:109--116, 2013.

\bibitem{KratschPPW15}
S.~Kratsch, M.~Pilipczuk, M.~Pilipczuk, and M.~Wahlstr{\"{o}}m.
\newblock Fixed-parameter tractability of {M}ulticut in directed acyclic
  graphs.
\newblock {\em {SIAM} J. Discrete Math.}, 29(1):122--144, 2015.

\bibitem{KratschW12}
S.~Kratsch and M.~Wahlstr{\"{o}}m.
\newblock Representative sets and irrelevant vertices: {N}ew tools for
  kernelization.
\newblock In {\em {FOCS 2012}}, pages 450--459. {IEEE} Computer Society, 2012.

\bibitem{LokshtanovMS11}
D.~Lokshtanov, D.~Marx, and S.~Saurabh.
\newblock Slightly superexponential parameterized problems.
\newblock In {\em {SODA} 2011}, pages 760--776, 2011.

\bibitem{LokshtanovRS16}
D.~Lokshtanov, M.~S. Ramanujan, and S.~Saurabh.
\newblock A linear time parameterized algorithm for {D}irected {F}eedback
  {V}ertex {S}et.
\newblock {\em CoRR}, abs/1609.04347, 2016.

\bibitem{lovasz-minimax}
L.~Lov\'asz.
\newblock On two minimax theorems in graph.
\newblock {\em J. Comb. Theory, Ser. B}, 21(2):96--103, 1976.

\bibitem{Lucchesi78aminimax}
C.~L. Lucchesi and D.~H. Younger.
\newblock A minimax theorem for directed graphs.
\newblock {\em J. London Math. Soc}, 17:369--374, 1978.

\bibitem{MarxP15}
D.~Marx and M.~Pilipczuk.
\newblock Optimal parameterized algorithms for planar facility location
  problems using {V}oronoi diagrams.
\newblock In {\em ESA 2015}, pages 865--877, 2015.

\bibitem{oeisA001003}
{OEIS Foundation Inc.}
\newblock {The On-Line Encyclopedia of Integer Sequences}.
\newblock \url{http://oeis.org/A001003}, 2017.

\bibitem{PilipczukW16}
M.~Pilipczuk and M.~Wahlstr{\"{o}}m.
\newblock Directed {M}ulticut is ${W}[1]$-hard, even for four terminal pairs.
\newblock In {\em {SODA 2016}}, pages 1167--1178. {SIAM}, 2016.

\bibitem{RobertsonS91}
N.~Robertson and P.~D. Seymour.
\newblock Graph minors. {X}. {O}bstructions to tree-decomposition.
\newblock {\em J. Comb. Theory, Ser. {B}}, 52(2):153--190, 1991.

\bibitem{SchrijverBook}
A.~Schrijver.
\newblock {\em Combinatorial Optimization -- Polyhedra and Efficiency}.
\newblock Springer, 2003.

\bibitem{SeymourT94}
P.~D. Seymour and R.~Thomas.
\newblock Call routing and the ratcatcher.
\newblock {\em Combinatorica}, 14(2):217--241, 1994.

\bibitem{stanley1997hipparchus}
R.~P. Stanley.
\newblock Hipparchus, {P}lutarch, {S}chr{\"o}der, and {H}ough.
\newblock {\em American Mathematical Monthly}, pages 344--350, 1997.

\end{thebibliography}

\end{document}